\documentclass[10pt]{article}

\usepackage[margin=0.7in]{geometry}

\usepackage{amsmath, amssymb, amsfonts, amsthm, mathtools}
\usepackage{xfrac}
\usepackage{bbm}

\usepackage{graphicx}
\usepackage{caption}
\usepackage[font=normalsize]{subcaption}
\usepackage{float}
\usepackage{booktabs}
\usepackage{longtable}
\usepackage{array}
\usepackage{tabularx}
\usepackage{multirow}
\usepackage{makecell}
\usepackage{adjustbox}
\usepackage{pifont}

\usepackage{algorithm}
\usepackage{algorithmic}

\usepackage{enumitem}
\usepackage{xcolor}
\usepackage{hyperref}

\usepackage{natbib}

\DeclareMathOperator{\Var}{Var}
\DeclareMathOperator{\Cov}{Cov}

\newtheorem{theorem}{Theorem}
\newtheorem{assumption}{Assumption}
\newtheorem{assumptions}{Assumptions}
\newtheorem{proposition}{Proposition}

\newtheorem{remark}{Remark}
\newtheorem{definition}{Definition}

\newenvironment{identifiabilityassumptions}{%
  \begin{assumptions}[Assumptions for Identifiability]%
}{%
  \end{assumptions}%
}

\newcolumntype{P}[1]{>{\raggedright\arraybackslash}p{#1}}

\newcommand\BibTeX{{\rmfamily B\kern-.05em \textsc{i\kern-.025em b}\kern-.08em
T\kern-.1667em\lower.7ex\hbox{E}\kern-.125emX}}

\begin{document}

\title{Causal Discovery via Statistical Power (CDSP)}

\author{
Shreya Prakash$^{1}$ \and
Fan Xia$^{4}$ \and
Elena A. Erosheva$^{1,2,3}$
}
\date{%
    $^1$Department of Statistics, University of Washington, Washington, USA\\%
    $^2$School of Social Work, University of Washington, Seattle, WA, USA\\
    $^3$Center for Statistics and the Social Sciences, University of Washington, Seattle, WA, USA\\
    $^4$Department of Epidemiology and Biostatistics, University of California San Francisco, California, USA\\
    [2ex]%
}
\maketitle

\abstract{Causal discovery methods aim to infer causal direction from observational data. Functional causal discovery approaches use structural asymmetries to identify causal directionality but rely on strong modeling assumptions and provide limited tools for uncertainty quantification. We introduce Causal Discovery via Statistical Power (CDSP), a statistical inference framework that connects causal direction estimation with statistical power and enables uncertainty quantification. Considering the foundational setting of bivariate observational data, we show how quantities analogous to statistical power and effect size can be used in causal discovery to determine when data contain sufficient information to favor one direction over the other. We introduce the effect-size asymmetry assumption that characterizes when the probability of correctly detecting the causal direction (i.e., the power of causal discovery) exceeds that of incorrectly favoring the reverse direction. We show that the effect-size asymmetry assumption can be used for causal direction estimation with uncertainty quantification. Simulations show that CDSP direction estimation is robust to mild and moderate model misspecifications. Real data analyses on 100 cause–effect benchmark pairs further demonstrate that CDSP reduces false discovery rates by approximately 18\% relative to a commonly used existing method.
}


\maketitle

\renewcommand\thefootnote{}
\footnotetext{Preprint.}

\renewcommand\thefootnote{\fnsymbol{footnote}}
\setcounter{footnote}{1}

\section{Introduction}

Inferring causal relationships from observational data is a central problem across many scientific disciplines \citep{glymour_review_2019, zhang2012inferring, saito2023causal, kotoku2020causal, rosenstrom2012pairwise, hu2018application, jiang2023signet,forastiere2021identification}. While causal effects can be established through randomized or interventional studies, such experiments are often expensive, time-consuming, or infeasible. As a result, there is substantial interest in methods that infer causal structure directly from observational data.

Functional causal discovery uses structural assumptions on the functional form to identify the causal direction. Functional causal discovery methods have been extensively studied in both bivariate and multivariate settings~\citep{shimizu2011directlingam, hoyer2009nonlinear, peters2014causal, zhang2012identifiability} and widely applied in practice~\citep{rosenstrom2012pairwise, motokawa2020causal, jiao2018bivariate, hu2018application, song2017tell}. 

In this paper, we focus on bivariate functional causal discovery that reduces to determining the causal direction between two variables by fitting regressions in both directions. The bivariate setting is important because it is the simplest setting in which causal discovery identifiability mechanisms can be studied. It can also be considered a building block for more complex multivariate settings. Bivariate causal discovery methods are often evaluated on benchmark datasets consisting of cause--effect pairs~\citep{mooij2016distinguishing}. 

In a bivariate setting, functional causal models posit that the data-generating process for two random variables \(X\) and \(Y\) and the true causal direction \(X \rightarrow Y\) can be written as
\begin{equation}
\label{eq:dgp_example}
   Y = f(X,\eta), 
\end{equation}
where \(\eta\) is a noise term that is independent of \(X\) and \(f \in \mathcal{F}\) is a function belonging to a specified model class.

Functional causal discovery methods rely on strong identifiability assumptions to identify the causal direction from observational data~\citep{shimizu2006linear, hoyer2009nonlinear}.
\begin{identifiabilityassumptions}\label{assump:model} In the bivariate setting, the identifiability assumptions are as follows:
\begin{enumerate}[label=(A\arabic*)]
    \item The data-generating mechanism is correctly specified within the assumed functional model class.
    \item $X$ and $Y$ have an acyclic relationship.
    \item $X$ and $Y$ are unconfounded, i.e., they share no common causes.
    \item The data are independent and identically distributed (i.i.d.).
    \item The data-generating process does not fall into the few unidentifiable cases characterized by \citet{zhang2012identifiability}; for example, when the relationship between $X$ and $Y$ is linear and all variables are Gaussian.
\end{enumerate}
\end{identifiabilityassumptions}

Under Assumptions~\ref{assump:model}, the independence between cause variable \(X\) and noise variable \(\eta\) holds only in the true direction and fails in the reverse direction. This asymmetry enables identification of a unique causal direction from observational data. When Assumptions~\ref{assump:model} hold, the causal direction can be identified; however, impacts of assumption violations on traditional causal discovery approaches are unclear. In addition, most existing causal discovery methods do not provide statistical guarantees---that is, statistical inference for the estimated direction.

In practice, assumption violations such as model misspecification due to deviations from the assumed functional form are often unavoidable. This has been widely recognized; for example, in linear regression, the consequences of assumption violations have been extensively studied, motivating the use of diagnostic tools and uncertainty quantification to assess model fit and guide interpretation~\citep{weisberg2005applied, kutner2005applied}.

In contrast, current work on uncertainty quantification in functional causal discovery has been limited. It can be broadly divided into two categories: resampling-based procedures that assess stability of estimated causal structures and inferential approaches based on statistical tests. Bootstrap-based methods extend causal discovery algorithms by repeatedly refitting models on resampled datasets to assess stability of inferred causal orderings \citep{komatsu2010assessing, thamvitayakul2012bootstrap}; however, as we will see in this paper, resulting reliability measures can be misleading when the underlying modeling assumptions are violated. \citet{wang2025confidence} develop statistical inference via confidence sets for causal orderings by inverting goodness-of-fit tests in identifiable structural equation models with additive errors. The confidence sets offer insights into model-class misspecification: an empty confidence set indicates that the ``model class does not capture the data-generating process''~\citep[p.~2]{wang2025confidence}. While this approach provides a formal method for uncertainty quantification for both bivariate and multivariate functional causal discovery, it does not extend to more general functional causal models such as the post-nonlinear model. Additionally, both resampling-based and inferential methods are developed under the assumption that the working models are correctly specified and may fail to provide reliable directional evidence under mild violations, as robustness to such assumption violations is not explicitly examined.

More broadly, much of the existing causal discovery literature emphasizes identifiability under idealized modeling assumptions, with comparatively limited attention paid to uncertainty quantification and robustness under assumption violations. A central challenge in applying causal discovery in practice is that identifiability assumptions---such as those related to model specification (Assumption~A1)---are rarely satisfied exactly, making it essential to develop methods that remain informative under such violations.

To address this challenge, we develop \emph{Causal Discovery via Statistical Power (CDSP)}, a framework that connects causal direction estimation with statistical power concepts. CDSP provides a statistical perspective on causal discovery by introducing notions of statistical power and effect size tailored to this setting. Within this framework, we introduce the \emph{effect-size asymmetry assumption}, which posits that the incorrect direction deviates more strongly from its null and is therefore easier to reject than the correct direction. We show that this asymmetry implies that the probability of detecting the correct causal direction exceeds that of detecting the incorrect direction, yielding a statistical characterization of directional evidence. We further leverage the effect-size asymmetry assumption to formulate causal direction estimation as a problem of assessing directional detectability. We develop the CDSP procedure to estimate the causal direction via effect-size asymmetry and to provide statistical inference for this estimate. Through simulation studies, we demonstrate that CDSP is robust to model misspecification. Analyzing 100 cause–effect benchmark pairs~\citep{mooij2016distinguishing}, we find that CDSP yields materially improved performance, reducing the false discovery rate by approximately 18\% relative to LiNGAM, both across all pairs and when restricted to pairs with approximately linear relationships---a setting in which methods based on the linear model class are expected to perform well.

Next, we introduce the CDSP framework and its theoretical foundations in Section~\ref{sec:methods_CDSP}, followed by numerical results in Section~\ref{sec:numerical_results}. We conclude with a discussion of our contributions and open questions for future research in Section~\ref{sec:discussion}.

\section{Causal Discovery via Statistical Power (CDSP)}
\label{sec:methods_CDSP}

In this section, we introduce Causal Discovery via Statistical Power (CDSP), a framework that connects causal direction estimation with statistical power. We begin by formulating a hypothesis-testing framework for bivariate causal discovery and defining a corresponding notion of directional power. We then show that this notion of power depends on the relative magnitudes of directional effect sizes, leading to the effect-size asymmetry assumption. Using asymptotic arguments, we establish that the effect-size asymmetry assumption is equivalent to favoring the correct direction with higher probability. Importantly, the effect-size asymmetry assumption does not explicitly depend on the presence of model misspecification. We then present the CDSP procedure, the implementation of the proposed framework for direction estimation, and discuss inference for the estimated direction through a measure of directional strength.

\subsection{Motivating the Role of Power, Effect Size, and Assumptions in Causal Discovery}
\label{subsec:motivate_subsamp} 

\begin{assumption}[Standing Assumption]
\label{assump:standing_xy}
Assume that the true causal direction is $X \to Y$, and the data-generating mechanism is
\begin{equation*}
    \label{eq:Y_gen}
Y = m(X) + \eta,
\end{equation*}
where $m$ is an arbitrary (unknown) generating function. The noise $\eta$ is, by construction, independent of $X$.

To assess the reverse direction $Y \to X$, we attempt to represent $X$ as a function of $Y$ by fitting a model of the form
\begin{equation*}
    \label{eq:X_eq}
X = l(Y) + \xi,
\end{equation*}
where $l$ is restricted to lie in a specified model class (e.g., linear functions under LiNGAM). The function $l$ is defined as the population risk minimizer,
\begin{equation*}
    \label{eq:l}
l = \arg\min_{h \in \mathcal{L}} R(h),
\end{equation*}
where $R(h)$ denotes the population risk. The reverse-direction residual is defined as
\begin{equation*}
    \label{eq:xi}
    \xi := X - l(Y).
\end{equation*}
Thus $(l,\xi)$ represents the best-fitting approximation to the true
data-generating process within the chosen model class when expressing the
data in the reverse direction.
\end{assumption}
Functional causal discovery methods assume that the regression function lies within the specified model class. Under this assumption, independence between the predictor and the noise term in the correct causal direction, but generally not in the incorrect one, creates an asymmetry that can be exploited for causal discovery. However, this independence cannot be tested directly because the noise term is unobserved. Instead, the noise terms, denoted by $\hat{\eta}$ and $\hat{\xi}$, must be estimated from regression models fit in the two possible directions. Measures of dependence between $\hat{\eta}$ and $X$, and between $\hat{\xi}$ and $Y$, are then compared to determine the causal direction. However, simply comparing the estimated dependence measures in the two directions does not account for the uncertainty in these empirical quantities. We therefore formulate causal direction detection as a pair of hypothesis tests corresponding to the two possible causal directions in \eqref{eq:h_indep},

\begin{equation}
    \label{eq:h_indep}
    \begin{aligned}
         H_Y^0: X \perp \eta,\; m \in \mathcal{M},
        \quad H_Y^1: \text{otherwise}, \\
        H_X^0: Y \perp \xi,\; l \in \mathcal{L},
        \quad H_X^1: \text{otherwise}.
\end{aligned}
\end{equation}

The null hypotheses $H_Y^0$ and $H_X^0$ jointly assess (i) independence between the predictor and residual and (ii) the adequacy of the model classes $\mathcal{M}$ and $\mathcal{L}$ for the $X \to Y$ and $Y \to X$ directions, respectively. Their alternatives, $H_Y^1$ and $H_X^1$, correspond to any violation of these conditions.

Here, $\mathcal{M}$ denotes the model class used to represent the generating mechanism $m$ in the $X \to Y$ direction, while $\mathcal{L}$ denotes the model class used to approximate the reverse-direction regression function $l$. In practice, these model classes are often taken to be the same (e.g., both linear), but they need not coincide in general.

Under the assumptions of functional causal discovery methods (Assumptions~\ref{assump:model}), the model class corresponding to the true causal direction is correctly specified\footnote{Model misspecification occurs when the true data-generating relationship does not follow the assumed functional form and therefore cannot be represented within the specified model class.}, the two tests are assessing only (i), and therefore the test results can directly inform causal direction. When either or both $\mathcal{M}$ or $\mathcal{L}$ are misspecified, rejection of a null hypothesis may arise from residual dependence, model misspecification, or both. This motivates decomposing each composite alternative into subcases that distinguish these possibilities in \eqref{eq:alt_decomp}.

\begin{equation}
\label{eq:alt_decomp}
\begin{aligned}
H_Y^{1a}:&\; X \not\!\perp \eta,\; m \in \mathcal{M},\\
H_Y^{1b}:&\; X \perp \eta,\; m \notin \mathcal{M},\\
H_Y^{1c}:&\; X \not\!\perp \eta,\; m \notin \mathcal{M},
\end{aligned}
\qquad
\begin{aligned}
H_X^{1a}:&\; Y \not\!\perp \xi,\; l \in \mathcal{L},\\
H_X^{1b}:&\; Y \perp \xi,\; l \notin \mathcal{L},\\
H_X^{1c}:&\; Y \not\!\perp \xi,\; l \notin \mathcal{L}.
\end{aligned}
\end{equation}

Intuitively, alternatives that depart more strongly from the null should be easier to detect and therefore yield higher power, whereas alternatives that represent weaker or more limited departures should generally yield lower power. In our setting, when model misspecification (Assumption~A1) is the sole source of violation, the correct causal direction satisfies the independence condition and can violate the null only through the model misspecification. Consequently, if the model class in the correct direction is only mildly misspecified, the corresponding alternative may remain close to the null and thus be harder to detect. This motivates studying power under the derived alternatives, as different forms of null violation can have different detectability and therefore provide evidence for causal direction.

In classical hypothesis testing \citep{cohen_power, lehmann2005testing, casella2021statistical}, the \emph{statistical power} of a test is defined as
\[
\mathbb{P}(\text{reject } H_0 \mid H_1 \text{ is true}),
\]
for a given null hypothesis $H_0$ and a given alternative hypothesis $H_1$.

Because causal discovery involves comparing two directional hypotheses, the notion of statistical power must be adapted to account for this paired testing structure. In particular, power corresponds to the probability that the procedure favors the correct direction.

\begin{definition}[Directional Power]
Directional power is the probability that the causal discovery method correctly favors $X \to Y$, namely
\begin{equation}
\label{eq:causal_power}
\mathbb{P}\Big(
\left(\text{fail to reject }H_Y^0\right)\wedge\left(\text{reject }H_X^0\right)
\,\Big|\, X\to Y
\Big).
\end{equation}
\end{definition}

To understand how this probability depends on the strength of directional evidence, we next introduce a directional effect size. Let $T_Y$ and $T_X$ denote the test statistics corresponding to the tests of $H_Y^0$ and $H_X^0$, respectively.
\begin{definition}[Directional Effect Size]
Let $\theta_Y$ and $\theta_X$ denote the population values of the test statistics $T_Y$ and $T_X$, respectively. Let $\sigma_Y/\sqrt{n}$ and $\sigma_X/\sqrt{n}$ denote their corresponding asymptotic standard deviations.

Define the directional effect size for testing $H_Y^0:\theta_Y=\theta_{0,Y}$ as
\[
\Delta_{Y,n}
=
\frac{(\theta_Y-\theta_{0,Y})\sqrt{n}}{\sigma_Y}.
\]
Similarly, the directional effect size for testing $H_X^0:\theta_X=\theta_{0,X}$ is
\[
\Delta_{X,n}
=
\frac{(\theta_X-\theta_{0,X})\sqrt{n}}{\sigma_X}.
\]

In our setting, we assume that under the null hypotheses of independence, the population values of the test statistics are zero (i.e., $\theta_{0,Y} = \theta_{0,X} = 0$). This assumption holds for commonly used dependence measures such as the Hilbert–Schmidt Independence Criterion (HSIC)~\citep{gretton_kernel_2008}. Under this assumption, the directional effect sizes simplify to
\begin{equation}
\label{eq:def_effect_size}
\Delta_{Y,n}
=
\frac{\theta_Y\sqrt{n}}{\sigma_Y}, 
\qquad
\Delta_{X,n}
=
\frac{\theta_X\sqrt{n}}{\sigma_X}.
\end{equation}
\end{definition}
\begin{assumption}[Traditional Asymmetry Assumption]
\label{assump:trad_asymmetry}
Traditionally, functional causal discovery methods rely on an asymmetry between the two directions, which can be expressed in terms of the population test statistics as
\[
\theta_X \;>\; \theta_Y.
\]

Under the standard identifiability assumptions (Assumptions~\ref{assump:model}), the correct direction $X \to Y$ yields $\theta_Y = 0$, while the reverse direction $Y \to X$ yields $\theta_X > 0$. Thus, this inequality arises as a direct consequence of the model assumptions when they are satisfied.
\end{assumption}

Although this Assumption~\ref{assump:trad_asymmetry} is not typically stated explicitly in the literature, it is implicitly relied upon. For example, many functional causal discovery methods compare test statistics or dependence measures across directions and select the direction with the smaller value (e.g., \citet{shimizu2006linear,mooij2016distinguishing,peters2017elements}).

Building on the traditional asymmetry assumption (Assumption~\ref{assump:trad_asymmetry}), we incorporate the variability of the test statistics through standardization and account for the corresponding critical values, so that the comparison reflects how strongly each direction departs from its null, yielding a refined comparison in Assumption~\ref{assump:effect_asymmetry}.

\begin{assumption}[Effect-Size Asymmetry]
\label{assump:effect_asymmetry}
Let $q_{Y,n}^\alpha$ and $q_{X,n}^\alpha$ denote the directional standardized critical values,
\begin{equation}
\label{eq:def_std_crit}
q_{Y,n}^\alpha
=
\frac{\sqrt{n}\, c_{Y,n}^\alpha}{\sigma_Y}, 
\qquad
q_{X,n}^\alpha
=
\frac{\sqrt{n}\, c_{X,n}^\alpha}{\sigma_X},
\end{equation}
where $c_{Y,n}^\alpha$ and $c_{X,n}^\alpha$ denote the level-$\alpha$ critical values under the null hypotheses $H_Y^0$ and $H_X^0$, respectively. We say that the effect-size asymmetry assumption holds if
\begin{equation}
    \label{eq:effect_size_asymm}
    \frac{\theta_X}{\sigma_X} > \frac{\theta_Y}{\sigma_Y} \quad \text{and} \quad q_{X,n}^\alpha - q_{Y,n}^\alpha = o(1)
\end{equation}
\end{assumption}

\begin{remark}
\label{rem:agnostic_model_miss}
The effect-size asymmetry assumption depends on the relative magnitudes of the directional effect sizes and thus does not explicitly depend on model misspecification.
\end{remark}

When model misspecification (Assumption~A1) is the sole source of violation, the alternatives in Equation~\eqref{eq:h_indep} simplify to those in Equation~\eqref{eq:test_alt_miss}. 

\begin{equation} 
\label{eq:test_alt_miss} 
\begin{aligned}
        H_Y^{1b}:\; X \perp \eta,\; m \notin \mathcal{M} 
        \;\text{ and }\; 
        H_X^{1c}:\; Y \not\!\perp \xi,\; l \notin \mathcal{L}, 
        & \text{ if $X \to Y$ is true}, \\[6pt]
        H_Y^{1c}:\; X \not\!\perp \eta,\; m \notin \mathcal{M} 
        \;\text{ and }\; 
        H_X^{1b}:\; Y \perp \xi,\; l \notin \mathcal{L}, 
        & \text{ if $Y \to X$ is true}. 
\end{aligned}
\end{equation}

In this regime, the incorrect direction ($Y \to X$) exhibits two sources of deviation from its null---both $Y \not\perp \xi$ and $l \notin \mathcal{L}$---whereas the correct direction ($X \to Y$) exhibits only one, namely $m \notin \mathcal{M}$. The effect-size asymmetry assumption posits that two such deviations induce a larger departure from the null than a single deviation, and hence a larger departure relative to the corresponding rejection threshold. This disparity yields asymmetric finite-sample behavior: the incorrect direction is more readily detectable and is therefore rejected more quickly than the correct direction. 

To make the connection between directional effect sizes and rejection probabilities concrete, we work under the misspecification regime in Equation~\eqref{eq:test_alt_miss} where the true causal direction is $X \to Y$, the correct-direction null is violated through model misspecification only (i.e., $H_Y^{1b}$), and the incorrect-direction null is violated through both dependence and model misspecification (i.e., $H_X^{1c}$). 

In what follows, we use the independence and goodness-of-fit test of \citet{sen_testing_2014}, which jointly assesses (i) independence between the predictor and residuals using the Hilbert--Schmidt independence criterion (HSIC)~\citep{gretton_kernel_2008} and (ii) whether the regression function lies within the assumed functional model class. We focus on linear model misspecification. 
Theorem~\ref{thm:bivariate_joint} in the Appendix shows that the corresponding test statistics $(T_{1b,Y},T_{1c,X})$ are asymptotically bivariate normal under the alternatives.\footnote{We present the results in Theorem~\ref{thm:effect_size_power} under an asymptotic Gaussian approximation for clarity, although the same arguments apply more generally to any pair of directional test statistics with a known joint distribution, whether asymptotic or finite-sample.}

We show in Theorem~\ref{thm:effect_size_power} that, under this regime, the effect-size asymmetry assumption is equivalent to the probability of correctly detecting the causal direction (i.e., directional power) exceeding the probability of incorrectly favoring the reverse direction.

\begin{theorem}[Effect-Size Asymmetry Characterization of Directional Power]
\label{thm:effect_size_power}
Consider the misspecification regime in which the correct-direction null is violated through model misspecification only (i.e., $H_Y^{1b}$ holds) while the incorrect-direction null is violated through both dependence and model misspecification (i.e., $H_X^{1c}$ holds).

Let $T_{1b,Y}$ and $T_{1c,X}$ be the test statistics corresponding to $H_Y^{1b}$ and $H_X^{1c}$. Let $c_Y^\alpha$ and $c_X^\alpha$ denote the level-$\alpha$ critical values for $nT_{1b,Y}$ and $nT_{1c,X}$, respectively \citep{sen_testing_2014}, and define $c_{Y,n}^\alpha := \frac{c_Y^\alpha}{n}$ and $c_{X,n}^\alpha := \frac{c_X^\alpha}{n}$. Define the directional detectability indices
\begin{equation}
    \label{eq:dir_idx}
    I_{Y,n}:=\frac{\theta_Y - c_{Y,n}^\alpha}{\sigma_Y},
    \qquad
    I_{X,n}:=\frac{\theta_X - c_{X,n}^\alpha}{\sigma_X}.
\end{equation}
Equivalently, $\sqrt{n}\, I_{Y,n} = \Delta_{Y,n} - q_{Y,n}^\alpha, \sqrt{n}\,I_{X,n} = \Delta_{X,n} - q_{X,n}^\alpha$. Denote by $\eta$ the noise term in the true direction and by $\xi$ the noise term in the reverse direction, as defined in Assumption~\ref{assump:standing_xy}. Define the misspecified residuals in the true and incorrect directions as
\[
\epsilon = m(X) - X \tilde{\beta}_0 + \eta, 
\qquad 
\delta = l(Y) - Y \tilde{\gamma}_0 + \xi,
\]
where $\tilde{\beta}_0$ and $\tilde{\gamma}_0$ are the population linear regression coefficients in the correct and incorrect directions, respectively. 
Let
\[
\theta_Y = \theta(X,\epsilon),
\qquad
\theta_X = \theta(Y,\delta),
\]
denote the population HSIC dependence measures corresponding to $(X,\epsilon)$ and $(Y,\delta)$, respectively, and let
\[
\sigma_Y^2 := \sigma_{1b,Y}^2,
\qquad
\sigma_X^2 := \sigma_{1c,X}^2,
\]
denote the asymptotic variances of $T_{1b,Y}$ and $T_{1c,X}$. Then
\begin{align*}
\sqrt{n}(T_{1b,Y}-\theta_Y) &\xrightarrow{d} \mathcal N(0,\sigma_Y^2),\\
\sqrt{n}(T_{1c,X}-\theta_X) &\xrightarrow{d} \mathcal N(0,\sigma_X^2).
\end{align*}
As a result, under the asymptotic normal approximations,
\[
\begin{aligned}
&\mathbb{P}\!\left((\text{reject } H_X^0)\wedge(\text{fail to reject } H_Y^0)\,\middle|\,H_Y^{1b},H_X^{1c}\right) \\
&\qquad >
\mathbb{P}\!\left((\text{reject } H_Y^0)\wedge(\text{fail to reject } H_X^0)\,\middle|\,H_Y^{1b},H_X^{1c}\right)
\end{aligned}
\]
if and only if
\begin{equation}
\label{eq:effect_short}
I_{X,n} > I_{Y,n},
\end{equation}
or equivalently,
\begin{equation}
\label{eq:effect_short_scaled}
\Delta_{X,n} - q_{X,n}^\alpha > \Delta_{Y,n} - q_{Y,n}^\alpha.
\end{equation}
Furthermore, since $q_{X,n}^\alpha - q_{Y,n}^\alpha = o(1)$, then asymptotically Equation~\ref{eq:effect_short} is equivalent to 
\[
\frac{\theta_X}{\sigma_X} > \frac{\theta_Y}{\sigma_Y}.
\]
Thus, directional effect sizes $\Delta_{X,n}$ and $\Delta_{Y,n}$ govern detectability of the true causal direction.
\end{theorem}

\begin{proof}
Under the misspecification regime $(H_Y^{1b},H_X^{1c})$, decompose the marginal rejection probabilities as
\begin{align*}
\mathbb{P}(\text{reject } H_X^0 \mid H_Y^{1b},H_X^{1c})
&=
\mathbb{P}\!\left((\text{reject } H_X^0)\wedge(\text{fail to reject } H_Y^0)\,\middle|\,H_Y^{1b},H_X^{1c}\right) \\
&\qquad +
\mathbb{P}\!\left((\text{reject } H_X^0)\wedge(\text{reject } H_Y^0)\,\middle|\,H_Y^{1b},H_X^{1c}\right),
\\[0.5em]
\mathbb{P}(\text{reject } H_Y^0 \mid H_Y^{1b},H_X^{1c})
&=
\mathbb{P}\!\left((\text{reject } H_Y^0)\wedge(\text{fail to reject } H_X^0)\,\middle|\,H_Y^{1b},H_X^{1c}\right) \\
&\qquad +
\mathbb{P}\!\left((\text{reject } H_X^0)\wedge(\text{reject } H_Y^0)\,\middle|\,H_Y^{1b},H_X^{1c}\right).
\end{align*}
Subtracting the two identities gives
\begin{align}
&\mathbb{P}\!\left((\text{reject } H_X^0)\wedge(\text{fail to reject } H_Y^0)\,\middle|\,H_Y^{1b},H_X^{1c}\right) \nonumber\\
&\qquad -
\mathbb{P}\!\left((\text{reject } H_Y^0)\wedge(\text{fail to reject } H_X^0)\,\middle|\,H_Y^{1b},H_X^{1c}\right) \nonumber\\
&=
\mathbb{P}(\text{reject } H_X^0 \mid H_Y^{1b},H_X^{1c})
-
\mathbb{P}(\text{reject } H_Y^0 \mid H_Y^{1b},H_X^{1c}).
\label{eq:directional_detectability_diff}
\end{align}

Because the rejection event for $H_X^0$ depends only on the test statistic $T_{1c,X}$ and the rejection event for $H_Y^0$ depends only on the test statistic $T_{1b,Y}$, the corresponding marginal rejection probabilities are determined by the marginal asymptotic laws of these statistics. Hence
\[
\mathbb{P}(\text{reject } H_X^0 \mid H_Y^{1b},H_X^{1c})
=
\mathbb{P}(\text{reject } H_X^0 \mid H_X^{1c}),
\]
and
\[
\mathbb{P}(\text{reject } H_Y^0 \mid H_Y^{1b},H_X^{1c})
=
\mathbb{P}(\text{reject } H_Y^0 \mid H_Y^{1b}).
\]
Therefore, Equation \eqref{eq:directional_detectability_diff} becomes
\begin{align}
&\mathbb{P}\!\left((\text{reject } H_X^0)\wedge(\text{fail to reject } H_Y^0)\,\middle|\,H_Y^{1b},H_X^{1c}\right) \nonumber\\
&\qquad -
\mathbb{P}\!\left((\text{reject } H_Y^0)\wedge(\text{fail to reject } H_X^0)\,\middle|\,H_Y^{1b},H_X^{1c}\right) \nonumber\\
&=
\mathbb{P}(\text{reject } H_X^0 \mid H_X^{1c})
-
\mathbb{P}(\text{reject } H_Y^0 \mid H_Y^{1b}).
\label{eq:directional_detectability_marginals}
\end{align}
Under the asymptotic normal approximation we have,
\begin{align*}
\mathbb{P}(\text{reject } H_X^0 \mid H_X^{1c})
&=
\mathbb{P}(nT_{1c,X}>c_X^\alpha \mid H_X^{1c}) \\
&=
\mathbb{P}\!\left(
\frac{\sqrt{n}(T_{1c,X}-\theta_X)}{\sigma_X}
>
\frac{\sqrt{n}(c_X^\alpha/n-\theta_X)}{\sigma_X}
\,\middle|\,H_X^{1c}
\right) \\
&\approx
1-\Phi\!\left(\frac{\sqrt{n}(c_X^\alpha/n-\theta_X)}{\sigma_X}\right) \\
&=
\Phi\!\left(\frac{\sqrt{n}(\theta_X-c_X^\alpha/n)}{\sigma_X}\right).
\end{align*}
Similarly,
\begin{align*}
\mathbb{P}(\text{reject } H_Y^0 \mid H_Y^{1b})
&=
\mathbb{P}(nT_{1b,Y}>c_{Y}^\alpha \mid H_Y^{1b}) \\
&=
\mathbb{P}\!\left(
\frac{\sqrt{n}(T_{1b,Y}-\theta_Y)}{\sigma_Y}
>
\frac{\sqrt{n}\left(c_Y^\alpha/n-\theta_Y\right)}{\sigma_Y}
\,\middle|\,H_Y^{1b}
\right) \\
&\approx
1-\Phi\!\left(\frac{\sqrt{n}(c_Y^\alpha/n-\theta_Y)}{\sigma_Y}\right) \\
&=
\Phi\!\left(\frac{\sqrt{n}(\theta_Y-c_Y^\alpha/n)}{\sigma_Y}\right).
\end{align*}
Substituting these expressions into Equation \eqref{eq:directional_detectability_marginals} yields
\begin{align*}
&\mathbb{P}\!\left((\text{reject } H_X^0)\wedge(\text{fail to reject } H_Y^0)\,\middle|\,H_Y^{1b},H_X^{1c}\right) \\
&\qquad -
\mathbb{P}\!\left((\text{reject } H_Y^0)\wedge(\text{fail to reject } H_X^0)\,\middle|\,H_Y^{1b},H_X^{1c}\right) \\
&=
\Phi\!\left(\frac{\sqrt{n}(\theta_X-c_X^\alpha/n)}{\sigma_X}\right)
-
\Phi\!\left(\frac{\sqrt{n}(\theta_Y-c_Y^\alpha/n)}{\sigma_Y}\right).
\end{align*}
Since $\Phi$ is strictly increasing, the right-hand side is positive if and only if
\begin{equation}
    \label{eq:ssc}
    \frac{\sqrt{n}(\theta_X-c_X^\alpha/n)}{\sigma_X} > \frac{\sqrt{n}(\theta_Y-c_Y^\alpha/n)}{\sigma_Y} \Longleftrightarrow \Delta_{X,n}-q_{X,n}^\alpha>\Delta_{Y,n}-q_{Y,n}^\alpha \Longleftrightarrow I_{X,n} > I_{Y,n}.
\end{equation}
Since $q_{X,n}^\alpha - q_{Y,n}^\alpha = o(1)$, then asymptotically Equation~\ref{eq:ssc} is equivalent to 
\[
\frac{\theta_X}{\sigma_X} > \frac{\theta_Y}{\sigma_Y}.
\]
This proves the claim.
\end{proof}

\begin{remark}
\label{rem:delta_q_interpretation}
The directional detectability indices $I_{X,n}$ and $I_{Y,n}$, defined in Equation~\ref{eq:dir_idx}, provide a finite-sample operational condition for evaluating effect-size asymmetry in terms of standardized effect sizes and critical values, and are what we use in practice for CDSP. As shown in Theorem~\ref{thm:effect_size_power}, for sufficiently large $n$, the condition in Equation~\ref{eq:ssc} is equivalent to the effect-size asymmetry assumption in Equation~\ref{eq:effect_size_asymm}. The latter provides a simpler, population-level characterization that aids intuition and interpretation, while the former captures the finite-sample quantities that govern the procedure.
\end{remark}

\begin{remark}
\label{rem:asymptotics_relevance}
CDSP is most informative under mild to moderate model misspecification, where the effect-size asymmetry is most likely to hold, as the test statistic corresponding to the correct direction tends to remain closer to the null than the test statistic corresponding to the incorrect direction. These settings correspond to regimes of greatest practical relevance, as severe model misspecification is often readily detectable using standard diagnostic tools (e.g., visual inspection), whereas mild to moderate misspecification is more subtle yet can still impact causal discovery.
\end{remark}

\subsection{Estimating Causal Direction and Directional Support Probability via CDSP}
\label{subsec:CDSP_procedure}

In Section~\ref{subsec:motivate_subsamp}, we defined notions of statistical power and effect size tailored to causal discovery and introduced the effect-size asymmetry assumption, which posits that the standardized deviation from the null is larger in the incorrect direction than in the correct direction. We showed that this asymmetry is equivalent to the probability of detecting the correct causal direction exceeding the probability of detecting the incorrect direction. In this section, we leverage effect-size asymmetry to estimate the causal direction and quantify the strength of directional evidence.

Suppose we have a bivariate dataset of size \(N\) on variables \(X\) and \(Y\). We assume that the true data-generating process is unknown (i.e., we do not know whether $X \to Y$ or $Y \to X$, nor the functional form of $f$ in Equation~\ref{eq:dgp_example}). We obtain CDSP point estimate of the causal direction and an estimate of the directional support probability using Algorithm~\ref{alg:CDSP_procedure}. Specifically, CDSP Procedure outputs (i) estimated causal direction $D_{\text{CDSP}}$: either $X \to Y$, $Y \to X$, or inconclusive, and (ii) directional support measure $P_{\text{CDSP}}$.

\begin{algorithm}
\caption{CDSP Procedure: Estimation and Inference via the Effect-Size Asymmetry}
\label{alg:CDSP_procedure}
\begin{algorithmic}[1]

\STATE \textbf{Input:} Dataset $\{(X_i,Y_i)\}_{i=1}^N$, significance level $\alpha$, number of bootstrap replicates $B$.

\STATE Using bootstrap, estimate $\hat\theta_Y$ and $\hat\theta_X$ as bootstrap means, $\hat\sigma_{Y}$ and $\hat\sigma_{X}$ as bootstrap standard deviations, and $\hat{c}_{Y,n}^\alpha$ and $\hat{c}_{X,n}^\alpha$ via the procedure of \citet{sen_testing_2014}.
\STATE Compute directional detectability index estimates
\[
\hat I_{Y,n} = \frac{\hat\theta_Y - \hat c_{Y,n}^\alpha}{\hat\sigma_Y},
\qquad
\hat I_{X,n} = \frac{\hat\theta_X - \hat c_{X,n}^\alpha}{\hat\sigma_X}.
\]
\STATE \textbf{Point estimate of direction:}
\[
\hat{D}_{\text{CDSP}}=
\begin{cases}
X \to Y, & \text{if } \hat I_{X,n}> \hat I_{Y,n},\\
Y \to X, & \text{if } \hat I_{Y,n} > \hat I_{X,n}.
\end{cases}
\]

\STATE \textbf{Bootstrap inference:} For $b = 1,\dots,B$, draw a bootstrap resample and recompute $\hat I_{X,n}^{(b)}$ and $\hat I_{Y,n}^{(b)}$ using Steps 2 and 3.

\STATE Estimate directional support probability:
\[
\widehat{P}_{\text{CDSP}} =
\begin{cases}
\frac{1}{B}\sum_{b=1}^B 
\mathbf{1}\{\hat I_{X,n}^{(b)} - \hat I_{Y,n}^{(b)} > 0\},
& \text{if } \hat{D}_{\text{CDSP}} = X \to Y,\\[6pt]
\frac{1}{B}\sum_{b=1}^B 
\mathbf{1}\{\hat I_{Y,n}^{(b)} - \hat I_{X,n}^{(b)} > 0\},
& \text{if } \hat{D}_{\text{CDSP}} = Y \to X.
\end{cases}
\]

\STATE \textbf{Output:} $\hat{D}_{\text{CDSP}}$ and $\widehat{P}_{\text{CDSP}}$.

\end{algorithmic}
\end{algorithm}

We estimate causal direction $D_{\text{CDSP}}$ in Algorithm~\ref{alg:CDSP_procedure} via effect-size asymmetry. By Theorem~\ref{thm:effect_size_power}, effect-size asymmetry is equivalent to the probability of detecting the correct causal direction exceeding the probability of favoring the incorrect direction. We first estimate the required components needed to evaluate effect-size asymmetry via the directional detectibility indices---$\theta_Y$ and $\theta_X$ via bootstrap means, $\sigma_{Y}$ and $\sigma_{X}$ via bootstrap standard deviations, and $c_{Y,n}^\alpha$ and $c_{X,n}^\alpha$ via the procedure of \citet{sen_testing_2014}---and then substitute these estimates into Equation~\eqref{eq:dir_idx}. We then compare $\hat I_{X,n}$ and $\hat I_{Y,n}$ to determine $\hat{D}_{\text{CDSP}}$: if $\hat I_{X,n} > \hat I_{Y,n}$, we set $\hat{D}_{\text{CDSP}} = X \to Y$; if $\hat I_{Y,n} > \hat I_{X,n}$, we set $\hat{D}_{\text{CDSP}} = Y \to X$.\footnote{In the event that $\hat I_{X,n} = \hat I_{Y,n}$, $\hat{D}_{\text{CDSP}}$ is inconclusive. This case is unlikely to occur in practice and is thus omitted from the procedure.}

The computational complexity of computing the point estimate in Algorithm~\ref{alg:CDSP_procedure} is \(O(BN^2)\). CDSP uses \(B\) bootstrap replicates to estimate the directional detectibility indices, and each replicate requires a constant number of HSIC computations. Under the standard empirical implementation, HSIC requires \(O(N^2)\) time due to pairwise kernel evaluations, yielding an overall time complexity of \(O(BN^2)\). The space complexity is \(O(N^2 + B)\), where \(O(N^2)\) arises from storing the HSIC kernel matrices and \(O(B)\) from storing the bootstrap outputs.

In addition to a point estimate of the causal direction, the CDSP procedure in Algorithm~\ref{alg:CDSP_procedure} outputs an inferential measure quantifying directional support. Specifically, if CDSP estimates \(X \to Y\), it reports the bootstrap estimate
\begin{equation}
\label{eq:p_CDSP}
\widehat{P}_{\mathrm{CDSP}}
=
\widehat{\mathbb{P}}\!\left(\hat{I}_{X,n}-\hat{I}_{Y,n}>0\right)
=
\frac{1}{B}\sum_{b=1}^B 
\mathbf{1}\!\left\{\hat I_{X,n}^{(b)} - \hat I_{Y,n}^{(b)} > 0\right\},
\end{equation}
where $B$ denotes the number of bootstrap samples, and 
$\hat I_{X,n}^{(b)}$ and $\hat I_{Y,n}^{(b)}$ are the corresponding estimates computed from the $b$-th resample. An analogous quantity is reported when CDSP estimates $Y \to X$. The computational complexity of bootstrap inference in Algorithm~\ref{alg:CDSP_procedure} is \(O(B^2 N^2)\), since it computes the point-estimation procedure across \(B\) bootstrap replicates, each of which has time complexity \(O(BN^2)\). The space complexity remains \(O(N^2 + B)\).

The population quantity $P_{\text{CDSP}}$ is the probability that the CDSP point estimate $\hat{D}_{\text{CDSP}}$ favors $X \to Y$ under sampling variability from the data-generating distribution. Thus, $\widehat{P}_{\mathrm{CDSP}}$ estimates the probability of directional support for $\hat{D}_{\text{CDSP}}$. From a decision-theoretic perspective, when the direction favored by CDSP is the true causal direction, $\widehat{P}_{\mathrm{CDSP}}$ can be interpreted as the success probability of the procedure under 0--1 loss---that is, the probability that the procedure selects the true causal direction \citep{casella2021statistical}.

Directional support probability $\widehat{P}_{\mathrm{CDSP}}$ can also be interpreted as a concordance probability, 
that is, the probability that one random quantity exceeds another. 
Probabilities of this form arise in several statistical contexts, including the area under the receiver operating characteristic curve (AUC) \citep{hanley1982meaning}.
In these settings, the quantity represents the probability that a randomly drawn observation from one group exceeds a randomly drawn observation from another. In the CDSP procedure, the randomness instead arises from the sampling distribution of the estimators $\hat{I}_{X,n}$ and $\hat{I}_{Y,n}$. Accordingly, $\widehat{P}_{\mathrm{CDSP}}$ measures the probability that the directional evidence favoring $X \to Y$ exceeds that favoring $Y \to X$ under sampling variability. Values of $\widehat{P}_{\mathrm{CDSP}}$ near $0.5$ indicate little directional separation between the two candidate directions, whereas values farther from $0.5$ indicate stronger directional dominance. Table~\ref{tab:pCDSP_interpretation} provides a heuristic interpretation scale adapted from AUC discrimination guidelines \citep{hosmer2013applied}.

\begin{table}[!ht]
\small
\centering
\begin{tabular}{ll}
\toprule
$\widehat{P}_{\mathrm{CDSP}}$ & Interpretation \\
\midrule
$\approx 0.5$ & Little or no directional separation \\
$0.5$--$0.7$ & Weak directional support \\
$0.7$--$0.8$ & Moderate directional support \\
$0.8$--$0.9$ & Strong directional support \\
$\ge 0.9$ & Very strong directional support \\
\bottomrule
\end{tabular}
\caption{Guidelines for interpretation of directional support $\widehat{P}_{\mathrm{CDSP}}$ values.}
\label{tab:pCDSP_interpretation}
\end{table}

\subsection{Establishing Consistency of CDSP Direction and Support Probability Estimates}
\label{subsec:CDSP_estimate}

In this section, we establish consistency of the inferred causal direction via effect-size asymmetry, and of the directional support probability $\widehat{P}_{\mathrm{CDSP}}$ from Algorithm~\ref{alg:CDSP_procedure}. In particular, we show that under the effect-size asymmetry assumption (Assumption~\ref{assump:effect_asymmetry}), the inferred causal direction is consistent.

\begin{theorem}[Consistency of Directional Estimate via Effect-Size Asymmetry]
\label{thm:ssc_consistent}
Let $c_{Y,n}^\alpha := c_Y^\alpha/n, c_{X,n}^\alpha := c_X^\alpha/n$, with plug-in estimates $\hat c_{Y,n}^\alpha := \hat c_Y^\alpha/n, \hat c_{X,n}^\alpha := \hat c_X^\alpha/n$. Suppose $\sigma_X$ and $\sigma_Y$ are bounded away from zero. Under standard regularity conditions, the estimators described in Algorithm~\ref{alg:CDSP_procedure} satisfy
\[
\hat{\sigma}_{X}\xrightarrow{p}\sigma_{X}, \quad
\hat{c}_{X}^{\alpha}\xrightarrow{p}c_{X}^{\alpha}, \quad
\hat{\theta}_X\xrightarrow{p}\theta_X,
\]
and
\[
\hat{\sigma}_{Y}\xrightarrow{p}\sigma_{Y}, \quad
\hat{c}_{Y}^{\alpha}\xrightarrow{p}c_{Y}^{\alpha}, \quad
\hat{\theta}_Y\xrightarrow{p}\theta_Y.
\]
Furthermore, let $I_X := \frac{\theta_X}{\sigma_X}$ and $I_Y := \frac{\theta_Y}{\sigma_Y}$. Then, if the asymptotic effect-size asymmetry assumption holds (Assumption~\ref{assump:effect_asymmetry}), that is,
\[
I_X > I_Y,
\]
then
\[
\mathbb{P}\bigl(\hat{I}_{X,n} > \hat{I} _{Y,n}\bigr) \to 1
\quad \text{as } n \to \infty.
\]
\end{theorem}
\begin{proof}
Prior work establishes consistency of the bootstrap estimators $\hat{c}_{X}^{\alpha}, \hat{c}_{Y}^{\alpha}, \hat{\theta}_X, \hat{\theta}_Y, \hat{\sigma}_{X},$ and $\hat{\sigma}_{Y}$ under standard regularity conditions (e.g., i.i.d.\ sampling and characteristic kernels) \citep{sen_testing_2014, gretton_kernel_2008, EfroTibs93}.

In particular, by the bootstrap consistency result of \citet{sen_testing_2014}, the estimated critical values satisfy
\begin{equation}
\label{eq:sen-critical}
\hat{c}_{X}^{\alpha} \xrightarrow{p} c^{\alpha}_{X}, 
\quad 
\hat{c}_{Y}^{\alpha} \xrightarrow{p} c^{\alpha}_{Y}.
\end{equation}
Since HSIC is a V\mbox{-}statistic, its empirical form is consistent \citep{gretton_kernel_2008}, implying
\begin{equation}
\label{eq:boot-mean}
\hat{\theta}_X = \hat{\theta}(Y,\delta) \xrightarrow{p} \theta(Y,\delta), 
\quad 
\hat{\theta}_Y = \hat{\theta}(X,\epsilon) \xrightarrow{p} \theta(X,\epsilon).
\end{equation}
Additionally, by standard bootstrap theory \citep{EfroTibs93},
\begin{equation}
\label{eq:boot-sd}
\hat{\sigma}_{X} \xrightarrow{p} \sigma_{X}, 
\quad 
\hat{\sigma}_{Y} \xrightarrow{p} \sigma_{Y}.
\end{equation}
Because $\sigma_X$ and $\sigma_Y$ are bounded away from zero, the continuous
mapping theorem~\citep{vandervaart1998asymptotic} gives
\[
\hat{I}_{X,n}
=
\frac{\hat\theta_X-\hat c_{X,n}^\alpha}{\hat\sigma_X}
\xrightarrow{p}
\frac{\theta_X}{\sigma_X} = I_X,
\qquad
\hat{I}_{Y,n}
=
\frac{\hat\theta_Y-\hat c_{Y,n}^\alpha}{\hat\sigma_Y}
\xrightarrow{p}
\frac{\theta_Y}{\sigma_Y} = I_Y.
\]
Here we use that
\[
c_{X,n}^\alpha=\frac{c_X^\alpha}{n}=o(1),
\qquad
c_{Y,n}^\alpha=\frac{c_Y^\alpha}{n}=o(1).
\]
By the asymptotic effect-size asymmetry assumption (Assumption~\ref{assump:effect_asymmetry}), we have $I_X > I_Y$. 
Let $\Delta:=I_X-I_Y>0$ and take $\varepsilon:=\Delta/4$. If $|\hat{I}_{X,n} - I_X| < \varepsilon$ and $|\hat{I}_{Y,n} - I_{Y}| < \varepsilon$, then
\[
\hat{I}_{X,n} > I_{X}-\varepsilon = I_{X}-\tfrac{\Delta}{4} \;>\; I_{Y}+\tfrac{\Delta}{4} = I_{Y}+\varepsilon > \hat{I}_{Y,n},
\]
hence $\hat{I}_{X,n}>\hat{I}_{Y,n}$. Therefore
\[
\{\hat{I}_{X,n}\le \hat{I}_{Y,n}\} \;\subseteq\; \{|\hat{I}_{X,n}-I_{X}|\ge \varepsilon\}\,\cup\,\{|\hat{I}_{Y,n}-I_{Y}|\ge \varepsilon\},
\]
and so
\[
\mathbb{P}(\hat{I}_{X,n}\le \hat{I}_{Y,n}) \;\le\; \mathbb{P}(|\hat{I}_{X,n}-I_{X}|\ge \varepsilon)\;+\;\mathbb{P}(|\hat{I}_{Y,n}-I_{Y}|\ge \varepsilon)\;\longrightarrow\;0.
\]
Thus $\mathbb{P}(\hat{I}_{X,n}>\hat{I}_{Y,n})\to 1$.
\end{proof}

For inference, we use the bootstrap estimate $\widehat{P}_{\mathrm{CDSP}}$. Proposition~\ref{prop:boot_prob_consistency} establishes consistency of $\widehat{P}_{\mathrm{CDSP}}$. In particular, when the effect-size asymmetry assumption holds for the true direction (e.g., $X \to Y$), the probability that CDSP selects $X \to Y$ converges to one as the sample size increases, while the probability of selecting $Y \to X$ converges to zero.

\begin{proposition}[Consistency of Bootstrap Directional Support Probability]
\label{prop:boot_prob_consistency}
Let $\mathcal{O}_n = \{(X_i, Y_i)\}_{i=1}^n$ denote an i.i.d.\ sample of size $n$, and let directional detectability indices $I_{X,n}$ and $I_{Y,n}$ be defined as in Equation~\eqref{eq:dir_idx}, with $D_n := \hat{I}_{X,n} - \hat{I}_{Y,n}$. By Theorem~\ref{thm:ssc_consistent}, $\mathbb{P}(D_n > 0) \to 1$, so that $D_n$ asymptotically reflects the true causal direction. Let $\mathcal{O}_n^*$ denote a nonparametric i.i.d.\ bootstrap resample drawn from $\mathcal{O}_n$, and let $D_n^*$ denote $D_n$ computed from $\mathcal{O}_n^*$. For the $b$-th bootstrap resample, let $D_n^{(b)} := \hat{I}_{X,n}^{(b)} - \hat{I}_{Y,n}^{(b)}$, and define
\[
\widehat{P}_{\mathrm{CDSP}}
= \frac{1}{B_n}\sum_{b=1}^{B_n}\mathbf{1}\!\left\{D_n^{(b)} > 0\right\},
\qquad B_n \to \infty.
\]

Define the (finite-$n$) directional support probability
\[
P_{\mathrm{CDSP},n} := \mathbb{P}(D_n>0).
\]

Assume the nonparametric i.i.d.\ bootstrap is asymptotically valid for $D_n$ in the sense that, conditionally on $\mathcal{O}_n$,
\[
\sup_{t\in\mathbb{R}}\bigg|
\mathbb{P}\!\left(D_n^* \le t \mid \mathcal{O}_n\right)
-
\mathbb{P}\!\left(D_n \le t\right)
\bigg|
\xrightarrow{p} 0.
\]

Then
\[
\widehat{P}_{\mathrm{CDSP}} - P_{\mathrm{CDSP},n} \xrightarrow{p} 0,
\qquad\text{and}\qquad
\widehat{P}_{\mathrm{CDSP}} \xrightarrow{p} 1.
\]
\end{proposition}

\begin{proof}
Let 
\[
P_n^* := \mathbb{P}\!\left(D_n^* > 0 \mid \mathcal{O}_n\right).
\]
Conditional on $\mathcal{O}_n$, the variables $\{\mathbf{1}(D_n^{(b)}>0)\}_{b=1}^{B_n}$ are i.i.d.\ Bernoulli$(P_n^*)$. Hence, by the conditional law of large numbers,
\[
\widehat{P}_{\mathrm{CDSP}} - P_n^* \xrightarrow{p} 0 \qquad (B_n \to \infty).
\]

It remains to show that $P_n^* - P_{\mathrm{CDSP},n} \xrightarrow{p} 0$. By bootstrap consistency,
\[
\mathbb{P}\!\left(D_n^* \le 0 \mid \mathcal{O}_n\right) - \mathbb{P}\!\left(D_n \le 0\right) \xrightarrow{p} 0,
\]
which implies
\[
P_n^* - P_{\mathrm{CDSP},n}
=
\mathbb{P}(D_n \le 0) - \mathbb{P}(D_n^* \le 0 \mid \mathcal{O}_n)
\xrightarrow{p} 0.
\]

Thus,
\[
\widehat{P}_{\mathrm{CDSP}} - P_{\mathrm{CDSP},n} \xrightarrow{p} 0.
\]

Finally, since $\mathbb{P}(D_n > 0) \xrightarrow 1$, we have $P_{\mathrm{CDSP},n} \to 1$, and therefore $\widehat{P}_{\mathrm{CDSP}} \xrightarrow{p} 1$.
\end{proof}

\begin{remark}
\label{rem:boot_dn}
Under the assumptions of Theorem~\ref{thm:ssc_consistent}, $D_n$ is a smooth plug-in functional of the empirical distribution,
as it is composed of consistent estimates with denominators bounded away from zero.
Under the nonparametric i.i.d.\ bootstrap, such smooth functionals are bootstrap-valid
under standard regularity conditions.
\end{remark}

\section{Numerical Results}
\label{sec:numerical_results}

In this section, we evaluate performance of CDSP through both simulated and real data experiments.

\subsection{Simulations}
\label{subsec:sims}

In our simulations, we compare the accuracy of CDSP and LiNGAM under increasing levels of linear model misspecification. 
To evaluate CDSP, we first assess whether the population effect-size asymmetry assumption holds using the ground truth quantities $I_X = \frac{\theta_X}{\sigma_X}$ and $I_Y = \frac{\theta_Y}{\sigma_Y}$.\footnote{Using the test statistics of \citet{sen_testing_2014}, the critical values satisfy $c^\alpha_{Y,n} = c^\alpha_{Y}/n$ and $c^\alpha_{X,n} = c^\alpha_{X}/n$. Consequently, the corresponding standardized critical values satisfy $q_{X,n}^\alpha-q_{Y,n}^\alpha = o(1)$, so asymptotically the contribution of the critical values is negligible, and population effect-size asymmetry is determined by $I_X$ and $I_Y$.} We then examine how effect-size asymmetry is reflected in finite samples via the distributions of the directional detectability indices $I_{X,n}$ and $I_{Y,n}$ across Monte Carlo replications. As misspecification increases, the overlap between the sampling distributions of $I_{X,n}$ and $I_{Y,n}$ grows, leading to a gradual decline in accuracy of the CDSP directional estimate and, under severe misspecification, a breakdown of the effect-size asymmetry assumption. In contrast, LiNGAM exhibits abrupt failure to detect the causal direction even under mild misspecification. These simulation results indicate that CDSP is more robust to misspecification than LiNGAM.

We generate bivariate datasets under the true causal direction $X \rightarrow Y$, considering settings that either satisfy the model assumptions or exhibit six increasing levels of linear model misspecification. Data are generated according to

\begin{equation}
    \label{eq:misspec_setting}
    Y = \mathrm{sign}(X-a)\,|X-a|^{d}\beta + \eta,
\end{equation}

\begin{figure}[!ht]
    \centering
    \includegraphics[width=0.5\textwidth]{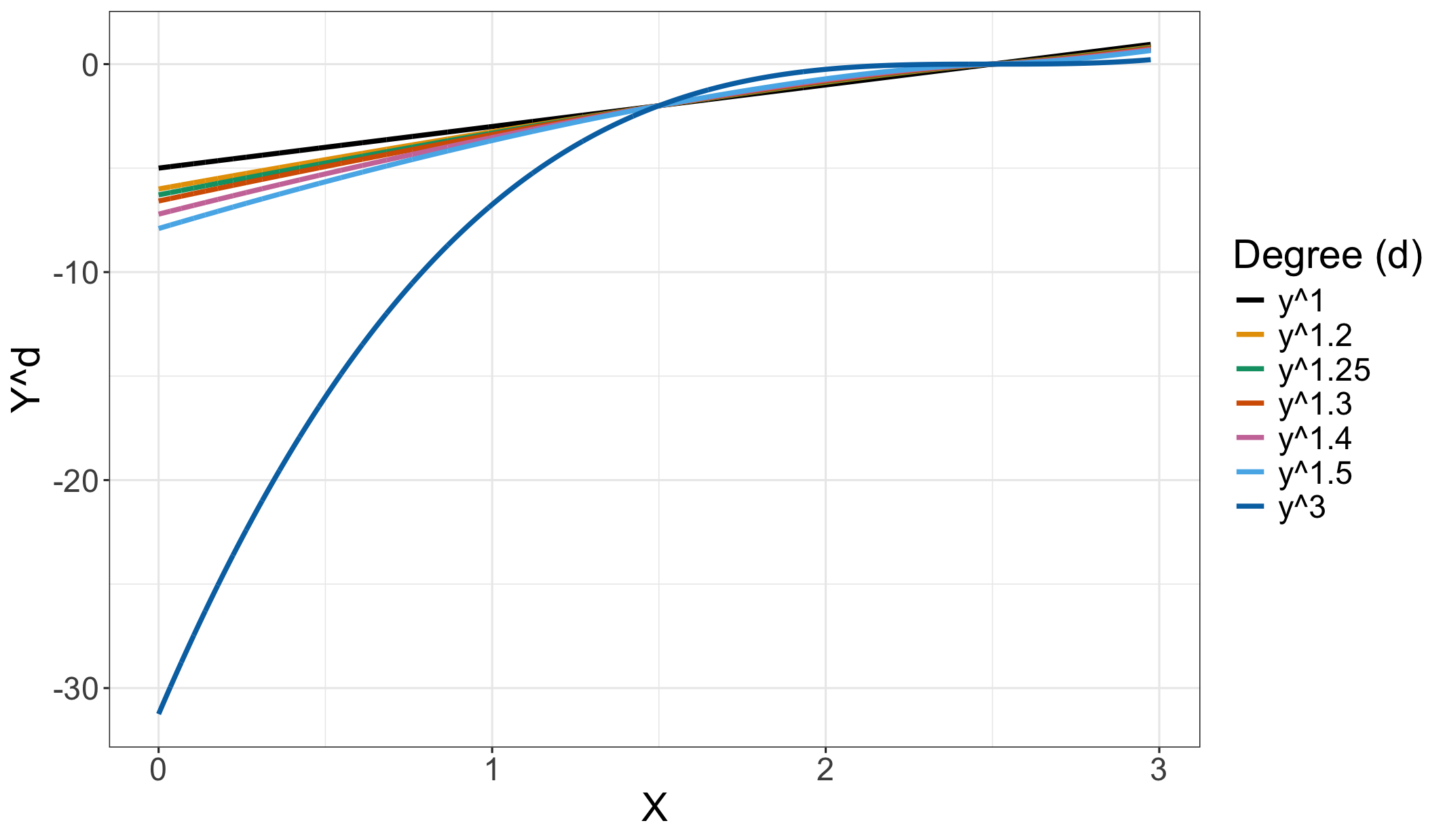}
    \caption{Settings of linearity
    }
    \label{fig:sim_polynomials}
\end{figure}

where the exponent $d$ controls the degree of nonlinearity. As illustrated in Figure~\ref{fig:sim_polynomials}, we consider $d \in \{1, 1.2, 1.25, 1.3, 1.4, 1.5, 3\}$ where $d = 1$ corresponds to the linear case and larger values represent progressively stronger departures from linearity. All other assumptions, aside from linearity, are satisfied. We sample $\eta$ from a Gaussian Mixture Model (GMM) with 3 mixtures. We sample $X$ from an exponential distribution truncated to the interval $(0, 3)$, ensuring that $Y$ is non-Gaussian. As seen in Figure~\ref{fig:sim_polynomials}, the case $d = 3$ represents a regime of severe model misspecification that lies outside the settings for which CDSP is primarily designed (Remark~\ref{rem:asymptotics_relevance}). Nonetheless, we discuss the case of $d = 3$ to highlight a regime in which the effect-size asymmetry assumption breaks down. 

For each nonlinearity level, we generate $M = 100$ datasets of size $N = 3000$ from the true data-generating process in Equation~\eqref{eq:misspec_setting}. For each dataset, we obtain a point estimate of the causal direction using CDSP (Algorithm~\ref{alg:CDSP_procedure}) with the test statistic of \citet{sen_testing_2014}, based on the Hilbert–Schmidt Independence Criterion (HSIC). We compare the CDSP point estimates of causal direction to those obtained from DirectLiNGAM \citep{shimizu2011directlingam}. To ensure comparability, we use HSIC as the independence measure within DirectLiNGAM.

To evaluate the population effect-size asymmetry assumption (Assumption~\ref{assump:effect_asymmetry}), we compute the required population quantities. We use a large Monte Carlo sample of size $10^5$ to estimate $\theta(X,\epsilon)$ and $\theta(Y,\delta)$. We estimate $\sigma_{1b,Y}$ and $\sigma_{1c,X}$ as $\sqrt{N}$ times the sample standard deviation of the test statistics $T_{1b,Y}$ and $T_{1c,X}$ across the $M=100$ Monte Carlo replications, where $N=3000$ is the sample size of each dataset. We compute accuracy of CDSP directional estimate as the empirical fraction of Monte Carlo replications in which the estimated ordering supports the direction favored by the population effect-size asymmetry assumption. We compute accuracy of DirectLiNGAM as the proportion of replications in which DirectLiNGAM correctly detects $X \to Y$. When the effect-size asymmetry assumption holds, accuracy of CDSP directional estimate reduces to the proportion of replications in which the correct direction is selected. Further details on the simulation setup are in the Appendix.

Figure~\ref{fig:CDSP_plots} illustrates the sampling distributions of directional detectability indices $I_{X,n}$ and $I_{Y,n}$ along with their corresponding population values ($I_X$ and $I_Y$), while Table~\ref{tab:pcdsp_linboot_compare} reports accuracies of the CDSP and LiNGAM directional estimates across mild to moderate misspecification settings $\left(d \in \{1, 1.2, 1.25, 1.3, 1.4, 1.5\}\right)$. In all cases, the effect-size asymmetry assumption holds, so both CDSP and LiNGAM accuracies are proportions of replications selecting the correct direction. 

\begin{table}[!ht]
\small
\centering
\begin{tabular}{lcc}
\toprule
& \multicolumn{2}{c}{\textbf{Accuracy}} \\
\cmidrule(lr){2-3}
Degree ($d$) & \textbf{CDSP} & \textbf{LiNGAM} \\
\midrule
$d=1$    & $100\%$ & $100\%$ \\
$d=1.2$  & $100\%$ & $100\%$ \\
$d=1.25$ & $100\%$ & $0\%$ \\
$d=1.3$  & $100\%$ & $0\%$ \\
$d=1.4$  & $88\%$  & $0\%$ \\
$d=1.5$  & $57\%$  & $0\%$ \\
\bottomrule
\end{tabular}
\caption{Accuracies of CDSP and LiNGAM directional estimates for polynomial degrees $d \in \{1, 1.2, 1.25, 1.3, 1.4, 1.5\}$.}
\label{tab:pcdsp_linboot_compare}
\end{table}

\begin{figure}[!ht]
    \centering
    \includegraphics[width=0.8\textwidth]{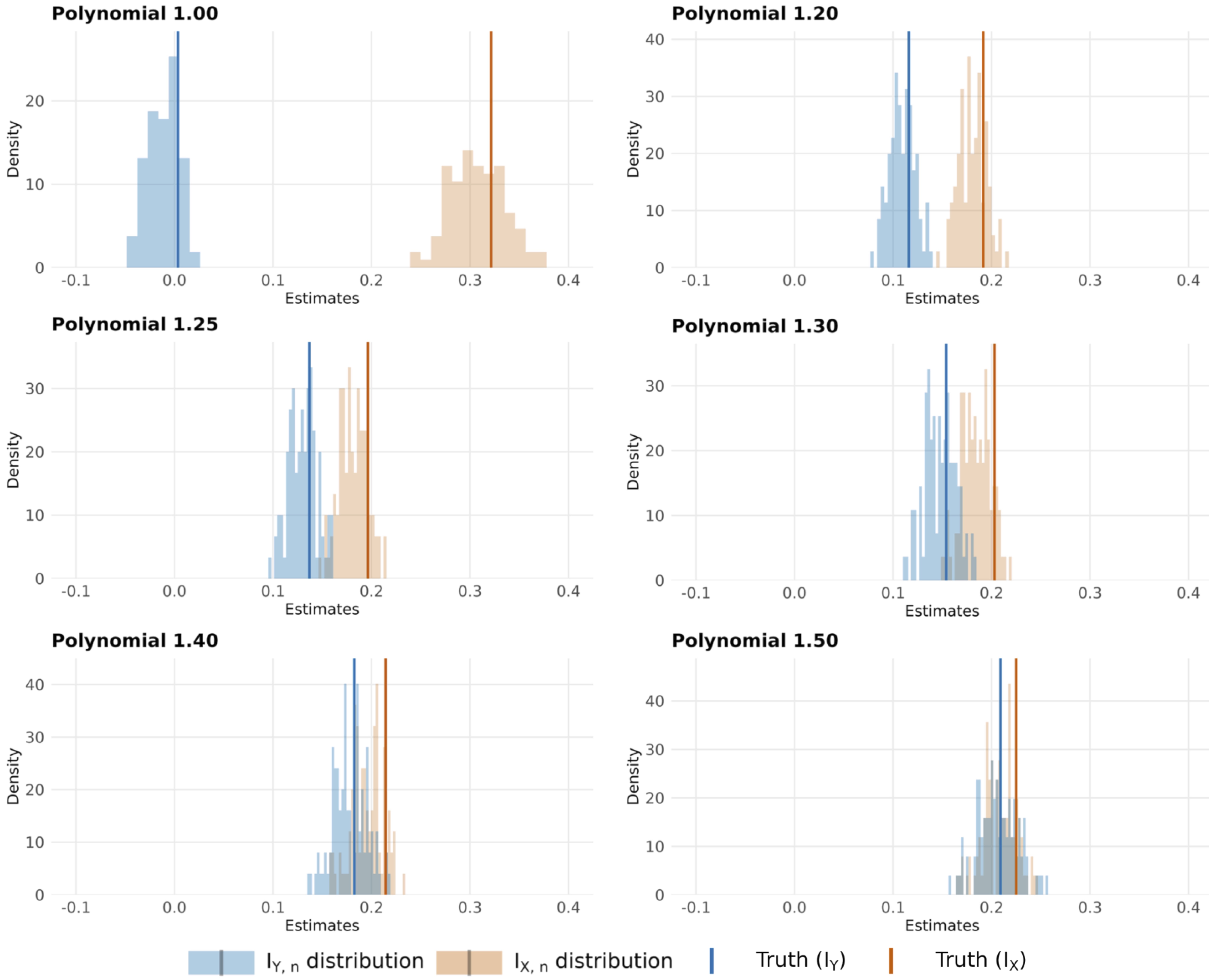}
    \caption{
    Distribution of the estimated directional detectability indices $I_{Y,n}$ (blue) and $I_{X,n}$ (orange) under varying degrees of nonlinearity. The top-left panel corresponds to the linear case ($d=1$). The remaining panels correspond to $d = 1.2$ (top-right), $1.25$ (middle-left), $1.3$ (middle-right), $1.4$ (bottom-left), and $1.5$ (bottom-right), representing progressively stronger departures from linearity. Solid blue and orange vertical lines indicate population values $I_{Y}$ and $I_{X}$, respectively. 
    }
    \label{fig:CDSP_plots}
\end{figure}

When linearity holds ($d = 1$, top left in Figure~\ref{fig:CDSP_plots} and first row of Table~\ref{tab:pcdsp_linboot_compare}), the sampling distributions of $I_{Y,n}$ and $I_{X,n}$ are well separated. In this setting, both CDSP and LiNGAM also correctly identify the direction across all $M$ replications. Because $H_Y^0$ holds, the corresponding population departure from the null, $\theta_Y$, is zero. In contrast, $H_X^0$ does not hold, so the corresponding population departure from the null, $\theta_X$, is strictly positive. This pronounced asymmetry induces a clear separation between the population values $I_X$ and $I_Y$, which is reflected in the well-separated sampling distributions of the directional detectability indices.

As the degree of misspecification increases from $d = 1$ to $d = 1.5$, the separation between the population values $I_X$ and $ I_Y$ decreases, and the sampling distributions of the directional detectability indices $I_{X,n}$ and $I_{Y,n}$ move closer together. Beginning at $d = 1.25$, the two distributions begin to overlap, and LiNGAM incorrectly identifies the direction across all replications, reflecting a breakdown of the traditional asymmetry assumption (Assumption~\ref{assump:trad_asymmetry}). Despite the slight overlap for $d \in \{1.25, 1.3\}$, CDSP accuracy remains at 100\%, indicating that the ordering $I_{X,n} > I_{Y,n}$ is preserved across Monte Carlo replications. This demonstrates that, under mild model misspecification, CDSP remains substantially more robust than LiNGAM.

At higher levels of misspecification $\left(d \in \{1.4, 1.5\}\right)$, the separation between the population values $I_X$ and $I_Y$ decreases further. The sampling distributions of the directional detectability indices $I_{X,n}$ and $I_{Y,n}$ exhibit substantial overlap across Monte Carlo replications, and accuracy of CDSP directional estimate decreases correspondingly. In these settings, sampling variability becomes large relative to the population gap between $I_{X}$ and $I_{Y}$, so the ordering $I_{X,n} > I_{Y,n}$ is no longer preserved uniformly across datasets. Nevertheless, CDSP continues to exhibit stronger directional performance than LiNGAM, which incorrectly concludes $Y \to X$ in all $M$ replications.

For the severe model misspecification case of $d = 3$, we observed that, unlike the previous settings $\left(d \in \{1, 1.2, 1.25, 1.3, 1.4, 1.5\}\right)$, the effect-size asymmetry assumption no longer holds. In this setting, CDSP continues to favor the direction implied by effect-size asymmetry; however, because this asymmetry is misaligned with the true causal direction, it results in incorrect inference relative to the ground truth $X \to Y$. As a result, both CDSP and LiNGAM favor the incorrect direction ($Y \to X$) across all Monte Carlo replications. We provide additional details for the $d=3$ case in the Appendix.

As discussed in Section~\ref{sec:methods_CDSP}, $\widehat{P}_{\text{CDSP}}$ measures the directional support for the estimated causal direction, defined as the proportion of bootstrap samples favoring the directional estimate. Similarly, LiNGAM bootstrap rates~\citep{thamvitayakul2012bootstrap} quantify the proportion of bootstrap resamples favoring the estimated direction. The accuracies of the CDSP and LiNGAM directional estimates reported in Table~\ref{tab:pcdsp_linboot_compare} are directly related to $\widehat{P}_{\text{CDSP}}$ and LiNGAM bootstrap rates, respectively: when the estimated direction is the true causal direction, $\widehat{P}_{\text{CDSP}}$ and the LiNGAM bootstrap rate correspond to accuracy, whereas when the estimated direction is incorrect, one minus the corresponding value reflects accuracy. For CDSP and LiNGAM, the effect-size asymmetry assumption and the traditional asymmetry assumption, respectively, characterize the theoretical conditions under which the estimated direction is expected to align with the true causal direction.

To assess the role of sample size, we increase $N$ from $3000$ to $6000$ and then to $9000$ in the $d \in \{1.4,1.5\}$ settings, where accuracy of CDSP directional estimate is below 100\%. LiNGAM is excluded from this study because its accuracy is consistently 0\% in these settings, indicating failure under this level of misspecification; as a result, increasing the sample size would not provide any additional insights into its performance. As shown in Table~\ref{tab:pCDSP_by_N}, for $d=1.4$, the accuracy increases from $88\%$ to $91\%$ and then to $94\%$, while for $d=1.5$ it increases from $57\%$ to $62\%$ and then to $65\%$. This behavior is consistent with Proposition~\ref{prop:boot_prob_consistency}: as the sample size increases, the estimates of the directional detectability indices $I_{Y,n}$ and $I_{X,n}$ concentrate more tightly around their population values $I_{Y}$ and $I_{X}$, reducing sampling variability and stabilizing their ordering across Monte Carlo replications. Consequently, when the effect-size asymmetry assumption holds, increasing the sample size can strengthen directional stability and improve accuracy of CDSP estimate. In settings where the effect-size asymmetry assumption holds but the directional signal is weak (e.g., $d=1.5$), this consistency may manifest in finite samples as an increasing trend in $\widehat{P}_{\mathrm{CDSP}}$ with sample size, as the ordering between $I_{Y,n}$ and $I_{X,n}$ stabilizes, thereby providing further evidence in favor of the correct direction.

\begin{table}[!ht]
\small
\centering
\begin{tabular}{lccc}
\toprule
Degree ($d$) & \textbf{$N=3000$} & \textbf{$N=6000$} & \textbf{$N=9000$} \\
\midrule
$d=1.4$ & 88\% & 91\% & 94\% \\
$d=1.5$ & 57\% & 62\% & 65\% \\
\bottomrule
\end{tabular}
\caption{Accuracy of CDSP directional estimate for increasing sample size $N$ for $d = \{1.4, 1.5\}$.}
\label{tab:pCDSP_by_N}
\end{table}

In summary, our simulations demonstrate CDSP's robustness to mild to moderate model misspecification and show that accuracy of CDSP directional estimate slightly declines as misspecification increases from mild to moderate. We additionally illustrate that CDSP outperforms LiNGAM across settings with varying degrees of nonlinearity. Importantly, the decline in CDSP accuracy reflects increased overlap in the sampling distributions of directional detectability indices $I_{Y,n}$ and $I_{X,n}$, rather than a structural failure of the method. We also demonstrate that increasing the sample size mitigates this effect: larger samples can strengthen directional stability when the effect-size asymmetry assumption holds. In contrast, LiNGAM fails to infer the correct direction even under mild misspecification, with accuracy dropping abruptly to 0\%.

\subsection{Real Data}
\label{subsec:real_data}

We apply CDSP as described in Algorithm~\ref{alg:CDSP_procedure} to the Tübingen Cause–Effect Pairs benchmark of 100 real-world univariate cause–effect pairs \citep{mooij2016distinguishing}, and compare the resulting direction estimates and associated uncertainty to those obtained from LiNGAM. For CDSP, we quantify uncertainty via the directional strength probability $\widehat{P}_{\mathrm{CDSP}}$. For LiNGAM, we use bootstrap selection rates as suggested by~\citep{thamvitayakul2012bootstrap} and commonly used in practice (e.g.,~\cite{motokawa2020causal,rosenstrom2012pairwise}).

As noted by the dataset’s authors, the Tübingen pairs are inherently heterogeneous: although each is believed to reflect a genuine causal relationship, the presumed ground-truth directions have not been established through intervention and hidden confounding or selection bias may be present~\citep{mooij2016distinguishing}. This heterogeneity is evident in the pairs' scatterplots that show a range of functional forms---from approximately linear relationships to strongly nonlinear ones. 

We consider two evaluation settings across pairs: the full set of 100 pairs and a subset of approximately linear pairs (i.e., pairs that do not exhibit clear violations of the linearity assumption). Given the substantial heterogeneity across all pairs, both LiNGAM and CDSP may struggle to reliably infer directionality on the full dataset. We expect better performance by both LiGNAM and CDSP on the subset of approximately linear pairs which reflects settings where both methods might reasonably be applied in practice. In addition, we further examine three pairs to better understand the comparative performance of the two procedures across different regimes: a pair where both CDSP and LiNGAM correctly identify the direction; a pair where LiNGAM selects the incorrect direction while CDSP correctly infers the direction with strong directional strength $\widehat{P}_{\mathrm{CDSP}}$; and a setting in which LiNGAM again selects the incorrect direction while CDSP correctly infers the direction but with weak directional strength.

To identify approximately linear relationships, we compare a linear regression model to a flexible spline-based generalized additive model (GAM) following standard practice in the GAM literature \citep{wood2017gam,ellison2022gam_bic}. For each pair $(X,Y)$ we fit (i) a linear regression model, $Y = \beta_0 + \beta_1 X + \varepsilon$ and (ii) a semiparametric GAM, $Y = \beta_0 + f(X) + \varepsilon$, where $f$ is a smooth function estimated using penalized regression splines with basis dimension $k=10$ \citep{wood2003thinplate,wood2017gam}. GAM Smoothing parameters are estimated using maximum likelihood to allow likelihood-based comparison across models~\citep{wood2017gam}. We compare the two models using the Bayesian Information Criterion (BIC) \citep{schwarz1978bic} and classify a pair as approximately linear when the linear model achieves lower BIC~\citep{kutner2005applied}.

We assess performance of LiNGAM and CDSP on the benchmark pairs using two metrics: the \emph{true discovery rate (TDR)} and \emph{false discovery rate (FDR)}. Following standard usage in statistical model evaluation \citep[e.g.,][]{powers2011evaluation, fawcett2006introduction}, we define
\[
\mathrm{TDR} = \frac{\text{number of correct decisions}}{\text{number of all decisions}},
\qquad
\mathrm{FDR} = \frac{\text{number of incorrect decisions}}{\text{number of all decisions}}.
\]

\begin{table}[!ht]
\small
\centering
\caption{Comparison of CDSP and LiNGAM on all 100 benchmark pairs. 
}
\label{tab:real_data_res}

\begin{subtable}{\textwidth}
\centering
\caption{Performance metrics across all 100 pairs.}
\begin{tabular}{@{}lll@{}}
\toprule
Metric                                  & CDSP          & LiNGAM        \\ \midrule
True Discovery Rate    & 62/100 = 62\% & 44/100 = 44\% \\
False Discovery Rate   & 38/100 = 38\% & 56/100 = 56\% \\
\bottomrule
\end{tabular}
\end{subtable}

\vspace{0.8em}

\begin{subtable}[t]{\textwidth}
\centering
\caption{Contingency table comparing causal discovery decisions across all 100 pairs.}
\begin{tabular}{@{}lccc@{}}
\toprule
                & LiNGAM Right & LiNGAM Wrong & Total \\ \midrule
CDSP Right        & 43           & 19            & 62    \\
CDSP Wrong        & 1            & 37           & 38    \\
Total             & 44          & 56           & 100    \\
\bottomrule
\end{tabular}
\end{subtable}

\begin{subtable}[t]{\textwidth}
\centering
\caption{False discovery rates by $P_{\text{CDSP}}$-defined support categories across all 100 pairs.}
\label{tab:fdr_support_all}
\begin{tabular}{@{}lcc@{}}
\toprule
Support Category & CDSP FDR & LiNGAM FDR \\ \midrule
No/little  & 2/3 = 66.7\%  & 2/8 = 25\% \\
Weak  & 18/45 = 40\% & 13/15 = 86.7\% \\
Moderate  & 4/8 = 50\%  & 7/10 = 70\% \\
Strong & 8/17 = 47.1\% & 6/10 = 60\% \\
Very strong & 6/27 = 22.2\% & 28/57 = 49.1\% \\ \midrule
Total & 38/100 = 38\% & 56/100 = 56\% \\
\bottomrule
\end{tabular}
\end{subtable}

\end{table}

\begin{table}[!ht]
\small
\centering
\caption{Comparison of CDSP and LiNGAM on 34 approximately linear pairs.}
\label{tab:real_data_sub}

\begin{subtable}[t]{\textwidth}
\centering
\caption{Performance metrics across a subset of 34 approximately linear pairs.}
\begin{tabular}{@{}lll@{}}
\toprule
Metric                                  & CDSP          & LiNGAM        \\ \midrule
True Discovery Rate  & 27/34 = 79.4\% & 21/34 = 61.8\% \\
False Discovery Rate & 7/34 = 20.6\% & 13/34 = 38.2\% \\
\bottomrule
\end{tabular}
\label{tab:linear_real_metrics}
\end{subtable}

\vspace{1em}

\begin{subtable}[t]{\textwidth}
\small
\centering
\caption{Contingency table comparing causal discovery decisions across 34 approximately linear pairs.}
\label{tab:sub-contingency}
\begin{tabular}{@{}lccc@{}}
\toprule
                & LiNGAM Right & LiNGAM Wrong & Total \\ \midrule
CDSP Right        & 21           & 6            & 27    \\
CDSP Wrong        & 0            & 7           & 7    \\
Total             & 21           & 13           & 34    \\
\bottomrule
\end{tabular}
\end{subtable}

\begin{subtable}[t]{\textwidth}
\small
\centering
\caption{False discovery rates by $P_{\text{CDSP}}$-defined support categories across 34 approximately linear pairs.}
\label{tab:fdr_support_linear}
\begin{tabular}{@{}lcc@{}}
\toprule
Support Category & CDSP FDR & LiNGAM FDR \\ \midrule
No/little& 1/1 = 100.0\% & 1/1 = 100.0\% \\
Weak& 4/9 = 44.4\% & 7/10 = 70.0\% \\
Moderate& 2/5 = 40.0\%  & 0/2 = 0.0\% \\
Strong& 0/7 = 0.0\%   & 2/4 = 50.0\% \\
Very strong& 0/12 = 0.0\%  & 3/17 = 17.6\% \\ \midrule
Total & 7/34 = 20.6\% & 13/34 = 38.2\% \\
\bottomrule
\end{tabular}
\end{subtable}
\end{table}

We present results across all 100 benchmark pairs in Table~\ref{tab:real_data_res}. Across the 100 real-data examples, CDSP achieves a substantially lower false discovery rate (38\% vs. 56\%) and a higher true discovery rate (62\% vs. 44\%). When the two methods disagree, CDSP is correct far more often: it is correct in 19 cases where LiNGAM is wrong, whereas LiNGAM is correct in only one case where CDSP is wrong. 

In Table~\ref{tab:fdr_support_all}, we stratify both $P_{\text{CDSP}}$ and LiNGAM bootstrap rates into the directional support categories defined in Table~\ref{tab:pCDSP_interpretation}. We use the $P_{\text{CDSP}}$-defined bins since no established interpretive scale exists for LiNGAM bootstrap rates. For a reliable method, we would expect false discovery rates to decrease as directional support increases, with low-support categories exhibiting higher error rates and high-support categories exhibiting lower error rates. We observe that LiNGAM-based bootstap rates provide `very strong' directional support to more than half of the pairs, yet LiNGAM exhibits a high false discovery rate within this category (49.1\%, compared to 22.2\% for CDSP). In contrast, CDSP exhibits higher false discovery rates in lower-support categories and substantially lower false discovery rates in higher-support categories. Thus, results in Table~\ref{tab:fdr_support_all} indicate that $P_{\text{CDSP}}$ provides a more meaningful measure of uncertainty quantification than LiNGAM's bootstrap rate.

Table~\ref{tab:real_data_sub} presents results for the approximately linear pairs. On this subset, as expected, classification performance improves for both methods; however, CDSP maintains a consistent advantage. CDSP again exhibits a lower false discovery rate (20.6\% vs. 38.2\%) and a higher true discovery rate (79.4\% vs. 61.8\%). In cases of disagreement, CDSP is correct 6 times, while LiNGAM is never correct when CDSP is wrong. In Table~\ref{tab:fdr_support_linear}, we again stratify both $P_{\text{CDSP}}$ and LiNGAM bootstrap rates into the directional support categories defined in Table~\ref{tab:pCDSP_interpretation}. We observe that false discovery rates are lower for both CDSP and LiNGAM in this subset compared to the full dataset. However, CDSP continues to show lower false discovery rates than LiNGAM, particularly in the higher-support categories, providing us with uncertainty quantification that is more meaningful than LiNGAM's bootstrap rates.

\begin{figure}[!ht]
    \centering
    \includegraphics[width=0.8\textwidth]{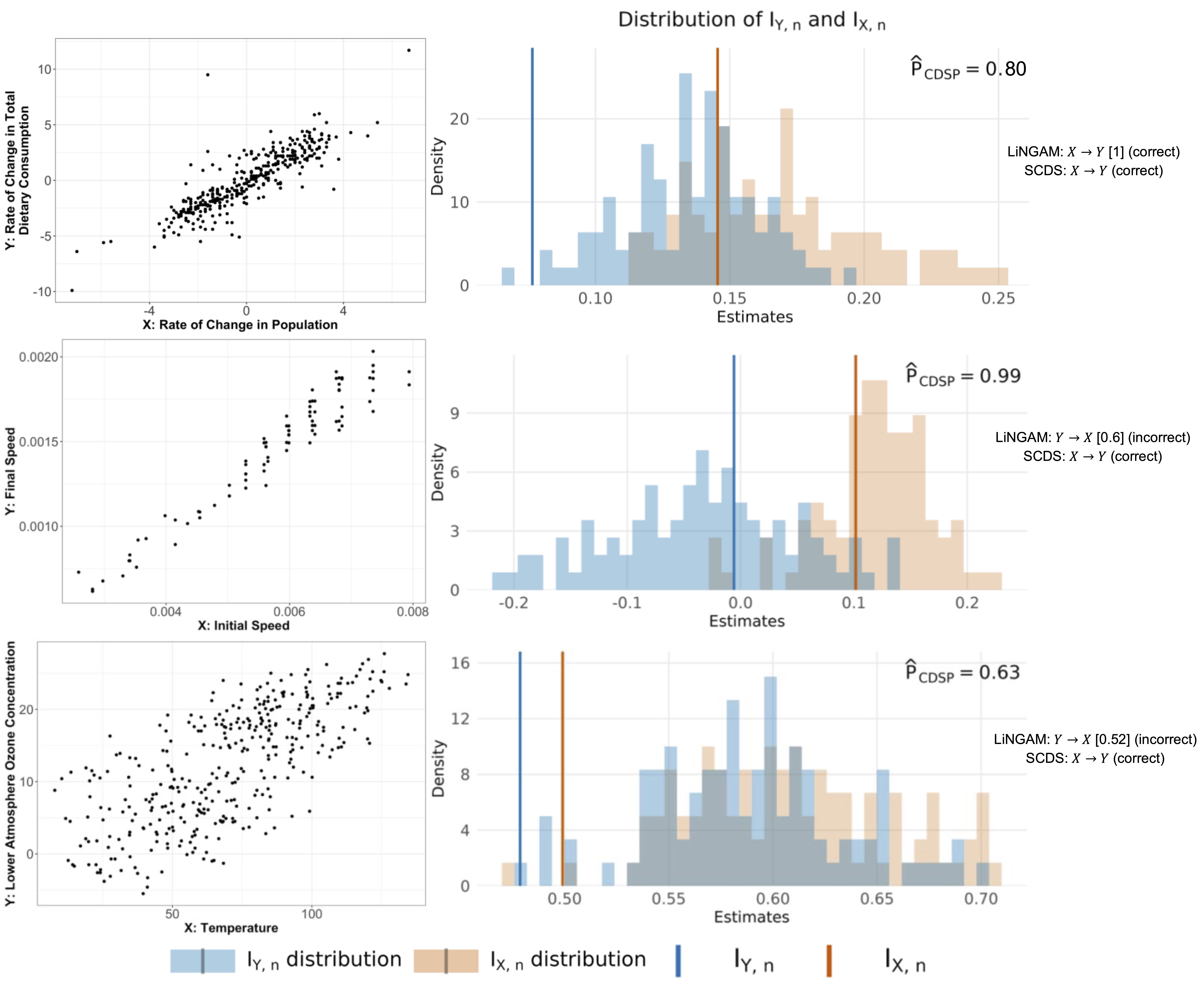}
    \caption{Rows 1–3 correspond to the Population and Dietary Consumption, Balltrack, and Ozone–Temperature datasets, respectively. The first column shows scatterplots of the observed data. The second column shows histograms of the bootstrap distributions of $I_{Y,n}$ (blue) and $I_{X,n}$ (orange); the solid vertical lines mark the corresponding estimates computed on the full dataset. Each histogram also reports $\widehat{P}_{\text{CDSP}}$. To the right of each histogram, we report the overall LiNGAM and CDSP results; LiNGAM bootstrap rates for the inferred direction are shown in brackets.}
    \label{fig:real_res}
\end{figure} 

In Figure~\ref{fig:real_res}, we examine three approximately linear pairs, selected to highlight distinct scenarios in which CDSP and LiNGAM differ in performance. The first pair relates population growth to food consumption and contains 347 observations from 174 countries or regions, covering two time periods: 1990–1992 to 1995–1997 and 1995–1997 to 2000–2002 (with one missing entry). The data, collected by the Food and Agriculture Organization of the United Nations, record the average annual rate of change in population ($X$) and total dietary consumption ($Y$). Following \citet{mooij2016distinguishing}, we treat the causal direction as $X \to Y$. The scatterplot in Figure~\ref{fig:real_res} suggests an approximately linear relationship between the two variables. In this well-behaved setting, both LiNGAM and CDSP recover the correct direction. CDSP yields $\widehat{P}_{\mathrm{CDSP}} = 0.80$, while LiNGAM yields a bootstrap rate of 1. As indicated in Table~\ref{tab:pCDSP_interpretation}, this value of $\widehat{P}_{\mathrm{CDSP}}$ corresponds to strong directional support for $X \to Y$.

The second pair relates the initial and final speeds of a rolling ball and consists of 94 measurements collected for the Tübingen Cause–Effect Pairs benchmark suite~\citep{mooij2016distinguishing}. The data were obtained using a physical ball track equipped with two light barrier sensors, where $X$ records the initial speed and $Y$ records the speed at a later position along the track. The causal direction is taken to be known, with initial speed causing final speed. The scatterplot (Figure~\ref{fig:real_res}) suggests an approximately linear relationship. Despite this apparent linearity, LiNGAM infers the incorrect direction, with a bootstrap rate of 0.6. In contrast, CDSP correctly identifies the direction, with ($\widehat{P}_{\mathrm{CDSP}} = 0.99$), corresponding to very strong support (Table~\ref{tab:pCDSP_interpretation}). Here, we observe a notable difference in variability between $X$ and $Y$, which may affect the variability of the corresponding test statistics in each direction. CDSP accounts for such variability in the effect-size asymmetry assumption. In contrast, LiNGAM does not explicitly account for this in the traditional asymmetry assumption, which may contribute to the discrepancy in performance, with CDSP correctly identifying the direction while LiNGAM does not.

The third pair relates ozone concentration to temperature and consists of 365 daily measurements of mean temperature and lower-atmosphere ozone concentration recorded in Lausanne-César-Roux, Switzerland in 2009~\citep{adminAirData}. Following \citet{mooij2016distinguishing}, we treat the causal direction as known, with changes in temperature ($X$) driving changes in ozone concentration ($Y$). The scatterplot (Figure~\ref{fig:real_res}) suggests an approximately linear relationship, making this a setting in which LiNGAM would ordinarily be expected to perform well. Despite this, LiNGAM infers the incorrect direction, with a bootstrap rate of 0.52. In contrast, CDSP correctly identifies the direction, with $\widehat{P}_{\mathrm{CDSP}} = 0.63$, indicating weak directional support (Table~\ref{tab:pCDSP_interpretation}). The weak support may in part be due to high variability in $X$, which can reduce power in the correct direction and diminish the separation between $\hat{I}_{Y,n}$ and $\hat{I}_{X,n}$. As suggested by Proposition~\ref{prop:boot_prob_consistency}, if the effect-size asymmetry assumption holds, increasing the sample size would be expected to strengthen the directional evidence and drive $\widehat{P}_{\mathrm{CDSP}}$ closer to one. Consistent with the patterns observed in Figure~\ref{fig:real_res}, the bootstrap distributions of the directional detectability index estimates $\hat{I}_{Y,n}$ and $\hat{I}_{X,n}$ are left-skewed, with the original-sample estimates appearing toward the lower tail. This suggests substantial variability in the sampling distributions, which can reduce the separation between $\hat{I}_{Y,n}$ and $\hat{I}_{X,n}$ and thereby contribute to the weak directional support. Another possible contributing factor is that the nonparametric bootstrap relies on an i.i.d.\ sampling assumption~\citep{EfroTibs93}, which may be imperfectly satisfied in this dataset, potentially affecting stability of the bootstrap estimates.

In sum, CDSP yields materially improved performance across both the full set of 100 benchmark pairs and the subset of 34 approximately linear pairs, reducing the false discovery rate by approximately 18\% in both settings relative to LiNGAM. We observe that LiNGAM bootstrap assigns very strong support to a large fraction of pairs, yet exhibits relatively high false discovery rates even within this category, particularly across all 100 pairs where assumption violations are more likely. In contrast, \(P_{\text{CDSP}}\) spans a broader range of values across support categories, with higher false discovery rates in lower-support categories and substantially lower false discovery rates in higher-support categories. These patterns indicate that larger values of \(\widehat{P}_{\mathrm{CDSP}}\) are associated with lower error rates, suggesting that \(\widehat{P}_{\mathrm{CDSP}}\) provides uncertainty quantification, especially under model misspecification. In addition, the three representative pairs illustrate how CDSP distinguishes between strong and weak evidence across settings, correctly inferring the causal direction with high support in settings where model assumptions are approximately satisfied and weaker support when the evidence is less clear.

\section{Discussion}
\label{sec:discussion}

In this work, we introduce Causal Discovery via Statistical Power (CDSP), a framework that provides a statistically principled basis for causal direction estimation. Existing methods often focus on identifiability under idealized modeling assumptions and pay limited attention to uncertainty quantification and robustness under model misspecification. In contrast, CDSP places these aspects at the center of the framework by characterizing when the data contain sufficient information to favor one causal direction over the other. CDSP reframes causal discovery as a problem of directional detectability through notions of power and effect size within a hypothesis-testing–based framework.

Through simulation studies involving linear model misspecification, we demonstrate that statistical inference via CDSP yields robust direction estimates and meaningful uncertainty quantification, while improving recovery of true causal direction relative to the commonly used LiNGAM. In addition, CDSP provides diagnostic insight into potential assumption violations through the underlying directional hypothesis test outcomes (Equation~\ref{eq:h_indep}). In particular, cases in which both directional hypotheses are either rejected or not rejected may indicate potential model misspecification or identifiability issues, respectively.

The most important practical consideration for using CDSP is that it relies on the effect-size asymmetry assumption (Assumption~\ref{assump:effect_asymmetry}), which represents a first step toward elucidating how effect size influences causal direction detection. This assumption cannot be directly verified in practice, since the true data-generating mechanism is unknown. However, it is weaker than the traditional asymmetry assumption (Assumption~\ref{assump:trad_asymmetry}) implicitly relied upon in functional causal discovery methods, which likewise cannot be verified in practice. The strength of the effect-size asymmetry depends on the choice of a test statistic, as different measures of dependence can capture deviations from the null to varying degrees. In this work, we demonstrate CDSP using the test statistic of \citet{sen_testing_2014}, which is based on the Hilbert--Schmidt independence criterion \citep{gretton_kernel_2008}. In our simulation studies, we observe that the effect-size asymmetry may fail under severe model misspecification but tends to hold under mild to moderate deviations. In practice, we therefore recommend applying standard model checking procedures appropriate to the assumed model class---e.g., linearity diagnostics for linear models~\citep{weisberg2005applied, kutner2005applied}---to avoid settings with severe misspecification, where the effect-size asymmetry assumption may not hold. Further evaluation of the effect-size asymmetry assumption---e.g., across different test statistics and more extensive simulation settings---remains an important direction for future work.

Our practical recommendation to ensure appropriateness of the assumed model class is confirmed by the real data analyses of the Tübingen cause–effect pairs~\citep{mooij2016distinguishing}. Thus, identifying 34 approximately linear pairs led to substantial performance improvements for both CDSP and LiNGAM. In particular, CDSP’s false discovery rate decreased by 17.4\% relative to its performance on the full set of 100 pairs. In addition, CDSP performed uniformly better than LiNGAM on real data, with its false discovery rate being approximately 18\% lower than LiNGAM's for both the full set of 100 cause-effect benchmark pairs and for the subset of 34 approximately linear pairs. We should note that our numerical findings on the Tübingen cause–effect pairs come with the caveat that, although widely treated as known for demonstration purposes in the causal discovery literature (e.g.,~\cite{peters2017elements, compton2020entropic, kocaoglu2020applications, marx2017telling}), the ground-truth causal directions are not known with complete certainty.

Another practical consideration in CDSP is the choice of resampling method. We use the nonparametric bootstrap to estimate the quantities required to infer causal directionality via effect-size asymmetry and \(P_{\mathrm{CDSP}}\), and establish consistency of these estimates in Section~\ref{subsec:CDSP_estimate}. Prior work has identified limitations of the nonparametric bootstrap in specific settings, such as when particular test statistics or model selection criteria are evaluated on bootstrap samples, or when combined with causal discovery followed by downstream causal effect estimation \citep{chang2026post, janitza2016pitfalls}. Even though these existing results pertain to different inferential targets, investigating how alternative resampling approaches affect CDSP performance and theoretical guarantees could be an important direction for future work.

We note that CDSP’s notion of ``directional power"---defined as the probability of correctly detecting the causal direction---is conceptually related to classical power in hypothesis testing, but is used here to characterize directional detectability rather than to perform classical power analysis~\citep{cohen_power, lehmann2005testing, casella2021statistical}. Our framework does not naturally lend itself to classical power analysis for causal discovery because such analyses would require specifying directional effect sizes, thereby assuming the causal direction---the very quantity the procedure seeks to infer.

In this paper, we illustrate CDSP in the context of linear model misspecification for bivariate causal discovery. However, CDSP applies more broadly to functional causal discovery methods, including ANMs and PNL models. In addition, beyond model misspecification, violations of other assumptions—such as unobserved confounding, cyclicity, or non-i.i.d.\ data—may also affect causal direction estimation. Understanding how such violations influence the effect-size asymmetry underlying CDSP remains an important direction for future work. In particular, studying how different types of assumption violations—individually or in combination—affect the effect-size asymmetry, and therefore the ability of CDSP to correctly identify the causal direction, is an open problem. Finally, in this work, we study CDSP in the bivariate setting. Extending the framework to higher-dimensional settings and developing scalable approaches for such applications are important directions for future research.

Overall, this paper takes a first step toward a power-motivated statistical framework for causal discovery. We encourage researchers use this method to learn when the data provide sufficient evidence to favor a causal direction, even in the presence of model misspecification.

\section*{Acknowledgments}
The authors thank participants of the University of Washington's Statistical and Machine Learning Approaches for the Social Sciences working group for their feedback and support. Shreya Prakash received partial support from the Center for Statistics and the Social Sciences (CSSS).

\bibliographystyle{plainnat}
\bibliography{causal-discovery}

@article{rosenstrom2012pairwise,
  title={Pairwise measures of causal direction in the epidemiology of sleep problems and depression},
  author={Rosenstr{\"o}m, Tom and Jokela, Markus and Puttonen, Sampsa and Hintsanen, Mirka and Pulkki-R{\aa}back, Laura and Viikari, Jorma S and Raitakari, Olli T and Keltikangas-J{\"a}rvinen, Liisa},
  journal={PloS one},
  volume={7},
  number={11},
  pages={e50841},
  year={2012},
  publisher={Public Library of Science San Francisco, USA}
}

@incollection{gretton_kernel_2008,
	title = {A {Kernel} {Statistical} {Test} of {Independence}},
	url = {http://papers.nips.cc/paper/3201-a-kernel-statistical-test-of-independence.pdf},
	urldate = {2019-05-14},
	booktitle = {Advances in {Neural} {Information} {Processing} {Systems} 20},
	publisher = {Curran Associates, Inc.},
	author = {Gretton, Arthur and Fukumizu, Kenji and Teo, Choon H. and Song, Le and Schölkopf, Bernhard and Smola, Alex J.},
	editor = {Platt, J. C. and Koller, D. and Singer, Y. and Roweis, S. T.},
	year = {2008},
	pages = {585--592}
}

@article{sen_testing_2014,
	title = {Testing independence and goodness-of-fit in linear models},
	volume = {101},
	issn = {0006-3444},
	url = {https://academic.oup.com/biomet/article/101/4/927/1775047},
	doi = {10.1093/biomet/asu026},
	language = {en},
	number = {4},
	urldate = {2020-10-04},
	journal = {Biometrika},
	author = {Sen, A. and Sen, B.},
	month = dec,
	year = {2014},
	note = {Publisher: Oxford Academic},
	pages = {927--942}
}

@article{song2017tell,
  title={Tell cause from effect: models and evaluation},
  author={Song, Jing and Oyama, Satoshi and Kurihara, Masahito},
  journal={International Journal of Data Science and Analytics},
  volume={4},
  pages={99--112},
  year={2017},
  publisher={Springer}
}

@article{shimizu2006linear,
  title={A linear non-Gaussian acyclic model for causal discovery},
  author={Shimizu, Shohei and Hoyer, Patrik O and Hyv{\"a}rinen, Aapo and Kerminen, Antti},
  journal={Journal of Machine Learning Research},
  volume={7},
  number={Oct},
  pages={2003--2030},
  year={2006}
}

@incollection{hoyer2009nonlinear,
	title = {Nonlinear causal discovery with additive noise models},
	url = {http://papers.nips.cc/paper/3548-nonlinear-causal-discovery-with-additive-noise-models.pdf},
	urldate = {2018-03-07},
	booktitle = {Advances in {Neural} {Information} {Processing} {Systems} 21},
	publisher = {Curran Associates, Inc.},
	author = {Hoyer, Patrik O. and Janzing, Dominik and Mooij, Joris M and Peters, Jonas and Schölkopf, Bernhard},
	editor = {Koller, D. and Schuurmans, D. and Bengio, Y. and Bottou, L.},
	year = {2009},
	pages = {689--696}
}

@article{mooij2016distinguishing,
  title={Distinguishing cause from effect using observational data: methods and benchmarks},
  author={Mooij, Joris M and Peters, Jonas and Janzing, Dominik and Zscheischler, Jakob and Sch{\"o}lkopf, Bernhard},
  journal={Journal of Machine Learning Research},
  volume={17},
  number={32},
  pages={1--102},
  year={2016}
}

@Book{EfroTibs93,
  Title                    = {An Introduction to the Bootstrap},
  Author                   = {Bradley Efron and Robert J. Tibshirani},
  Publisher                = {Chapman \& Hall/CRC},
  Year                     = {1993},

  Address                  = {Boca Raton, Florida, USA},
  Number                   = {57},
  Series                   = {Monographs on Statistics and Applied Probability}
}

@article{shimizu2011directlingam,
  title={{DirectLiNGAM}: A direct method for learning a linear non-Gaussian structural equation model},
  author={Shimizu, Shohei and Inazumi, Takanori and Sogawa, Yasuhiro and Hyvarinen, Aapo and Kawahara, Yoshinobu and Washio, Takashi and Hoyer, Patrik O and Bollen, Kenneth and Hoyer, Patrik},
  journal={Journal of Machine Learning Research},
  volume={12},
  number={Apr},
  pages={1225--1248},
  year={2011}
}

@inproceedings{thamvitayakul2012bootstrap,
  title={Bootstrap confidence intervals in {DirectLiNGAM}},
  author={Thamvitayakul, Kittitat and Shimizu, Shohei and Ueno, Tsuyoshi and Washio, Takashi and Tashiro, Tatsuya},
  booktitle={2012 IEEE 12th International Conference on Data Mining Workshops},
  pages={659--668},
  year={2012},
  organization={IEEE}
}

@article{compton2020entropic,
  title={Entropic causal inference: Identifiability and finite sample results},
  author={Compton, Spencer and Kocaoglu, Murat and Greenewald, Kristjan and Katz, Dmitriy},
  journal={Advances in Neural Information Processing Systems},
  volume={33},
  pages={14772--14782},
  year={2020}
}

@article{kocaoglu2020applications,
  title={Applications of common entropy for causal inference},
  author={Kocaoglu, Murat and Shakkottai, Sanjay and Dimakis, Alexandros G and Caramanis, Constantine and Vishwanath, Sriram},
  journal={Advances in neural information processing systems},
  volume={33},
  pages={17514--17525},
  year={2020}
}

@inproceedings{marx2017telling,
  title={Telling cause from effect using MDL-based local and global regression},
  author={Marx, Alexander and Vreeken, Jilles},
  booktitle={2017 IEEE international conference on data mining (ICDM)},
  pages={307--316},
  year={2017},
  organization={IEEE}
}

@book{weisberg2005applied,
  title     = {Applied Linear Regression},
  author    = {Weisberg, Sanford},
  year      = {2005},
  edition   = {3},
  publisher = {Wiley},
  address   = {Hoboken, NJ}
}

@book{kutner2005applied,
  title     = {Applied Linear Statistical Models},
  author    = {Kutner, Michael H. and Nachtsheim, Christopher J. and Neter, John and Li, William},
  year      = {2005},
  edition   = {5},
  publisher = {McGraw-Hill/Irwin},
  address   = {New York}
}

@article{jiao2018bivariate,
  title={Bivariate causal discovery and its applications to gene expression and imaging data analysis},
  author={Jiao, Rong and Lin, Nan and Hu, Zixin and Bennett, David A and Jin, Li and Xiong, Momiao},
  journal={Frontiers in genetics},
  volume={9},
  pages={347},
  year={2018},
  publisher={Frontiers Media SA}
}

@article{peters2014causal,
  title   = {Causal Discovery with Continuous Additive Noise Models},
  author  = {Peters, Jonas and Mooij, Joris M. and Janzing, Dominik and Sch{\"o}lkopf, Bernhard},
  journal = {Journal of Machine Learning Research},
  volume  = {15},
  pages   = {2009--2053},
  year    = {2014},
  url     = {https://jmlr.org/papers/v15/peters14a.html}
}

@book{lehmann2005testing,
  title     = {Testing Statistical Hypotheses},
  author    = {Lehmann, Erich L. and Romano, Joseph P.},
  year      = {2005},
  edition   = {3rd},
  publisher = {Springer}
}

@article{powers2011evaluation,
  title={Evaluation: From Precision, Recall and {F}-Measure to {ROC}, Informedness, Markedness \& Correlation},
  author={Powers, David M. W.},
  journal={Journal of Machine Learning Technologies},
  volume={2},
  number={1},
  pages={37--63},
  year={2011}
}

@article{janitza2016pitfalls,
  author  = {Janitza, Silke and Binder, Harald and Boulesteix, Anne-Laure},
  title   = {Pitfalls of hypothesis tests and model selection on bootstrap samples: causes and consequences in biometrical applications},
  journal = {Biometrical Journal},
  volume  = {58},
  number  = {3},
  pages   = {447--473},
  year    = {2016},
  doi     = {10.1002/bimj.201400246}
}

@article{chang2026post,
  title={Post-selection inference for causal effects after causal discovery},
  author={Chang, Ting-Hsuan and Guo, Zijian and Malinsky, Daniel},
  journal={Biometrika},
  volume={113},
  number={1},
  pages={asaf073},
  year={2026},
  publisher={Oxford University Press}
}

@article{fawcett2006introduction,
  title={An Introduction to {ROC} Analysis},
  author={Fawcett, Tom},
  journal={Pattern Recognition Letters},
  volume={27},
  number={8},
  pages={861--874},
  year={2006},
  publisher={Elsevier}
}

@book{vandervaart1998asymptotic,
  title        = {Asymptotic Statistics},
  author       = {van der Vaart, Aad W.},
  year         = {1998},
  publisher    = {Cambridge University Press},
  address      = {Cambridge},
  note         = {Continuous Mapping Theorem, Section 2.3}
}

@book{peters2017elements,
  title        = {Elements of Causal Inference: Foundations and Learning Algorithms},
  author       = {Peters, Jonas and Janzing, Dominik and Sch{\"o}lkopf, Bernhard},
  year         = {2017},
  publisher    = {MIT Press},
  address      = {Cambridge, MA},
  isbn         = {9780262037310}
}

@book{casella2021statistical,
  title     = {Statistical Inference},
  author    = {Casella, George and Berger, Roger L.},
  year      = {2021},
  edition   = {3rd},
  publisher = {Cengage Learning}
}

@article{wang2025confidence,
  title={Confidence sets for causal orderings},
  author={Wang, Y Samuel and Kolar, Mladen and Drton, Mathias},
  journal={Journal of the American Statistical Association},
  pages={1--14},
  year={2025},
  publisher={Taylor \& Francis}
}

@article{forastiere2021identification,
  title={Identification and estimation of treatment and interference effects in observational studies on networks},
  author={Forastiere, Laura and Airoldi, Edoardo M and Mealli, Fabrizia},
  journal={Journal of the American Statistical Association},
  volume={116},
  number={534},
  pages={901--918},
  year={2021},
  publisher={Taylor \& Francis}
}

@article{glymour_review_2019,
	title = {Review of {Causal} {Discovery} {Methods} {Based} on {Graphical} {Models}},
	volume = {10},
	issn = {1664-8021},
	url = {https://www.frontiersin.org/articles/10.3389/fgene.2019.00524},
	urldate = {2023-03-26},
	journal = {Frontiers in Genetics},
	author = {Glymour, Clark and Zhang, Kun and Spirtes, Peter},
	year = {2019},
}

@article{jiang2023signet,
  title={SIGNET: transcriptome-wide causal inference for gene regulatory networks},
  author={Jiang, Zhongli and Chen, Chen and Xu, Zhenyu and Wang, Xiaojian and Zhang, Min and Zhang, Dabao},
  journal={Scientific Reports},
  volume={13},
  number={1},
  pages={19371},
  year={2023},
  publisher={Nature Publishing Group UK London}
}

@article{wood2003thinplate,
  title={Thin plate regression splines},
  author={Wood, Simon N.},
  journal={Journal of the Royal Statistical Society: Series B (Statistical Methodology)},
  volume={65},
  number={1},
  pages={95--114},
  year={2003},
  publisher={Royal Statistical Society}
}

@book{wood2017gam,
  title={Generalized Additive Models: An Introduction with R},
  author={Wood, Simon N.},
  edition={2},
  year={2017},
  publisher={Chapman and Hall/CRC},
  address={Boca Raton}
}

@article{schwarz1978bic,
  title={Estimating the dimension of a model},
  author={Schwarz, Gideon},
  journal={The Annals of Statistics},
  volume={6},
  number={2},
  pages={461--464},
  year={1978},
  publisher={Institute of Mathematical Statistics}
}

@article{hanley1982meaning,
  title={The meaning and use of the area under a receiver operating characteristic (ROC) curve},
  author={Hanley, James A and McNeil, Barbara J},
  journal={Radiology},
  volume={143},
  number={1},
  pages={29--36},
  year={1982}
}

@inproceedings{komatsu2010assessing,
  title={Assessing statistical reliability of {LiNGAM} via multiscale bootstrap},
  author={Komatsu, Yusuke and Shimizu, Shohei and Shimodaira, Hidetoshi},
  booktitle={International Conference on Artificial Neural Networks},
  pages={309--314},
  year={2010},
  organization={Springer}
}

@article{motokawa2020causal,
  title        = {Causal relationships among health checkup variables with the use of {LiNGAM}},
  author       = {Motokawa, Tsukasa and Watanabe, Takuya and Moriyama, Koji and Takada, Yukio and Kitamura, Akihiro and Kato, Takahiro and Kokubo, Yoshihiro and Miyamoto, Yutaka and Okamura, Tomonori},
  journal      = {PLOS ONE},
  volume       = {15},
  number       = {12},
  pages        = {e0243432},
  year         = {2020},
  publisher    = {Public Library of Science},
  doi          = {10.1371/journal.pone.0243432},
  url          = {https://pmc.ncbi.nlm.nih.gov/articles/PMC7757823/}
}

@book{hosmer2013applied,
  title={Applied Logistic Regression},
  author={Hosmer, David W and Lemeshow, Stanley and Sturdivant, Rodney X},
  edition={3},
  publisher={Wiley},
  year={2013}
}

@article{ellison2022gam_bic,
  title={Model selection for generalized additive models using BIC},
  author={Ellison, Aaron M.},
  journal={Demographic Research},
  volume={47},
  pages={561--594},
  year={2022}
}

@article{zhang2012inferring,
  title={Inferring gene regulatory networks from gene expression data by path consistency algorithm based on conditional mutual information},
  author={Zhang, Xiujun and Zhao, Xing-Ming and He, Kun and Lu, Le and Cao, Yongwei and Liu, Jingdong and Hao, Jin-Kao and Liu, Zhi-Ping and Chen, Luonan},
  journal={Bioinformatics},
  volume={28},
  number={1},
  pages={98--104},
  year={2012},
  publisher={Oxford University Press}
}

@misc{adminAirData,
  author = {{Federal Office for the Environment (FOEN)}},
  title = {Air: Data, Indicators and Maps},
  year = {2009},
  url = {https://www.bafu.admin.ch/bafu/en/home/topics/air/state/data.html},
  note = {Accessed: 2026-04-21}
}

@article{saito2023causal,
  title={Causal analysis of nitrogen oxides emissions process in coal-fired power plant with {LiNGAM}},
  author={Saito, Tatsuki and Fujiwara, Koichi},
  journal={Frontiers in Analytical Science},
  volume={3},
  pages={1045324},
  year={2023},
  publisher={Frontiers}
}

@article{kotoku2020causal,
  title={Causal relations of health indices inferred statistically using the {DirectLiNGAM} algorithm from big data of Osaka prefecture health checkups},
  author={Kotoku, Jun’ichi and Oyama, Asuka and Kitazumi, Kanako and Toki, Hiroshi and Haga, Akihiro and Yamamoto, Ryohei and Shinzawa, Maki and Yamakawa, Miyae and Fukui, Sakiko and Yamamoto, Keiichi and others},
  journal={Plos one},
  volume={15},
  number={12},
  pages={e0243229},
  year={2020},
  publisher={Public Library of Science San Francisco, CA USA}
}

@article{hu2018application,
  title={Application of causal inference to genomic analysis: advances in methodology},
  author={Hu, Pengfei and Jiao, Rong and Jin, Li and Xiong, Momiao},
  journal={Frontiers in Genetics},
  volume={9},
  pages={238},
  year={2018},
  publisher={Frontiers Media SA}
}

@book{cohen_power,
  author    = {Jacob Cohen},
  title     = {Statistical Power Analysis for the Behavioral Sciences},
  edition   = {2nd},
  publisher = {Routledge},
  year      = {1988},
  address   = {New York},
  isbn      = {978-0-8058-0283-2},
  doi       = {10.4324/9780203771587}
}

@article{zhang2012identifiability,
  title={On the identifiability of the post-nonlinear causal model},
  author={Zhang, Kun and Hyvarinen, Aapo},
  journal={arXiv preprint arXiv:1205.2599},
  year={2012}
}

\appendix

\section{Asymptotic Joint Distribution of $(T_{1b,Y,n},\,T_{1c,X,n})$}
\label{app:joint_theorem}

In this appendix we show that the joint distribution of the test statistics
\((T_{1b,Y,n},\, T_{1c,X,n})\) is asymptotically bivariate normal as stated in
Section~\ref{subsec:motivate_subsamp}, adopting the independence–and–goodness-of-fit
framework of \citet{sen_testing_2014}. We work in the scalar case \(X\in\mathbb{R}\),
\(Y\in\mathbb{R}\), which is the setting used throughout our paper.

\subsection*{Setup and notation}

The true data-generating mechanism is
\begin{equation}
    Y = m(X) + \eta,
    \qquad E[\eta \mid X] = 0.
    \label{eq:gen_y_joint}
\end{equation}

Assume \(m\) is bijective on its image. For any measurable function \(l(Y)\), define
\[
\xi = m^{-1}(Y-\eta)-l(Y).
\]
For simplicity, let $l(Y)=E[X\mid Y],$ so that \(E[\xi\mid Y]=0\).

Let \((X_1,Y_1),\ldots,(X_Y,n_n)\) be i.i.d.\ observations generated from
\eqref{eq:gen_y_joint}. For \(i=1,\ldots,n\), define
\begin{enumerate}
    \item \(e_i = Y_i - X_i\hat{\beta}_n\) (residual in the true direction),
    \item \(\epsilon_i = m(X_i)-X_i\tilde{\beta}_0+\eta_i\) (misspecification error in the true direction),
    \item \(d_i = X_i - Y_i\hat{\gamma}_n\) (residual in the incorrect direction),
    \item \(\delta_i = l(Y_i)-Y_i\tilde{\gamma}_0+\xi_i\) (misspecification error in the incorrect direction).
\end{enumerate}

Let \(\theta(U,V)\) denote the Hilbert--Schmidt independence criterion \citep{gretton_kernel_2008}:
\begin{equation}
\label{eq:hsic_joint}
\theta(U,V)
=
E[k(U,U')h(V,V')]
+
E[k(U,U')]E[h(V,V')]
-
2E[k(U,U')h(V,V'')],
\end{equation}
where \((U',V')\) and \((U'',V'')\) are independent copies of \((U,V)\).

We consider the test statistics
\begin{align}
    T_{1b,Y,n}
    &= \frac{1}{n^2} \sum_{i,j=1}^n k_{ij} h_{ij}
    + \frac{1}{n^4} \sum_{i,j,q,r=1}^n k_{ij} h_{qr}
    - \frac{2}{n^3} \sum_{i,j,q=1}^n k_{ij} h_{iq},
    \label{eq:T1b_joint}
    \\
    T_{1c,X,n}
    &= \frac{1}{n^2} \sum_{i,j=1}^n k^{\prime}_{ij} h^{\prime}_{ij}
    + \frac{1}{n^4} \sum_{i,j,q,r=1}^n k^{\prime}_{ij} h^{\prime}_{qr}
    - \frac{2}{n^3} \sum_{i,j,q=1}^n k^{\prime}_{ij} h^{\prime}_{iq},
    \label{eq:T1c_joint}
\end{align}
where
\[
k_{ij}=k(X_i,X_j),\qquad h_{ij}=h(e_i,e_j),
\qquad
k'_{ij}=k'(Y_i,Y_j),\qquad h'_{ij}=h'(d_i,d_j),
\]
and \(k,k',h,h'\) are characteristic kernels defined on \(\mathbb{R}\times\mathbb{R}\).

Under \((H_Y^{1b},H_X^{1c})\), the corresponding population targets are
\[
\theta(X,\epsilon)
\qquad\text{and}\qquad
\theta(Y,\delta).
\]

Let
\[
T_n=
\begin{pmatrix}
T_{1b,Y,n}\\
T_{1c,X,n}
\end{pmatrix}.
\]

\subsection*{Conditions}

\textit{Condition 1.}
\[
A=E[X^2]>0,
\qquad
B=E[Y^2]>0.
\]

\textit{Condition 2.}
The kernels \(k,k',h,h'\) are characteristic, \(k,k'\) are continuous, and
\(h,h'\) are twice continuously differentiable. Their first- and second-order
partial derivatives are Lipschitz continuous.

\textit{Condition 3.}
\[
E[X^2]<\infty,\quad E[Y^2]<\infty,\quad E[\eta^2]<\infty,\quad E[\xi^2]<\infty, \quad E[X^2\epsilon^2]<\infty,\quad  \text{and} \quad E[Y^2\delta^2]<\infty.
\]
\textit{Condition 4.}
The kernel--moment conditions required by Lemma~1 in the Supplement of
\citet{sen_testing_2014} hold for both \((X,\epsilon)\) and \((Y,\delta)\), so that
the Taylor expansion remainders below satisfy
\[
R_n=o_p(n^{-1/2}),
\qquad
R_n'=o_p(n^{-1/2}).
\]

\textit{Condition 5.}
Both
\[
m(X)-X\tilde{\beta}_0
\qquad\text{and}\qquad
l(Y)-Y\tilde{\gamma}_0
\]
are nonconstant with positive probability.

\begin{theorem}
\label{thm:bivariate_joint}
Suppose that Conditions 1--5 hold. Then under \((H_Y^{1b},H_X^{1c})\),
\[
\sqrt{n}(T_n-\theta)\xrightarrow{d}N_2(0,\Sigma),
\]
where
\[
\theta=
\begin{pmatrix}
\theta(X,\epsilon)\\
\theta(Y,\delta)
\end{pmatrix},
\qquad
\Sigma=
\begin{pmatrix}
\sigma^2_{1b,Y} & \sigma_{1b,1c;X,Y}\\
\sigma_{1b,1c;X,Y} & \sigma^2_{1c,X}
\end{pmatrix},
\]
and
\[
\theta(X,\epsilon)>0,
\qquad
\theta(Y,\delta)>0.
\]
Moreover,
\begin{equation}
\label{eq:var_1bY}
\sigma^2_{1b,Y}
=
\Var\!\left(h^{(0)}_{Y,1}(W_1)+\nu A^{-1}X_1\epsilon_1\right),
\end{equation}
\begin{equation}
\label{eq:cov_1b1c}
\sigma_{1b,1c;X,Y}
=
\Cov\!\left(
h^{(0)}_{Y,1}(W_1)+\nu A^{-1}X_1\epsilon_1,\,
h^{(0)}_{X,1}(W'_1)+\phi B^{-1}Y_1\delta_1
\right),
\end{equation}
and
\begin{equation}
\label{eq:var_1cX}
\sigma^2_{1c,X}
=
\Var\!\left(h^{(0)}_{X,1}(W'_1)+\phi B^{-1}Y_1\delta_1\right),
\end{equation}
where \(W_i=(X_i,\epsilon_i)\), \(W_i'=(Y_i,\delta_i)\), and
\(h^{(0)}_{Y,1},h^{(0)}_{X,1},\nu,\phi\) are defined in the proof.
\end{theorem}

\begin{proof}
The least squares estimate
\[
\hat{\beta}_n=A_n^{-1}\left(\frac{1}{n}\sum_{i=1}^n X_iY_i\right),
\qquad
A_n=\frac{1}{n}\sum_{i=1}^n X_i^2,
\]
admits the following expansion around \(\tilde{\beta}_0\):
\begin{align}
\sqrt{n}(\hat{\beta}_n-\tilde{\beta}_0)
&=
\sqrt{n}\left(
A_n^{-1}\frac{1}{n}\sum_{i=1}^n X_iY_i-\tilde{\beta}_0
\right)
\nonumber\\
&=
\sqrt{n}\left(
A_n^{-1}\frac{1}{n}\sum_{i=1}^n X_i\{m(X_i)-X_i\tilde{\beta}_0+\eta_i\}
\right)
\nonumber\\
&=
\{1+o_p(1)\}\frac{1}{\sqrt{n}}\sum_{i=1}^n A^{-1}X_i\epsilon_i.
\label{eq:beta_exp_joint}
\end{align}
Here we use that \(A_n\to A\) almost surely and \(A>0\) by Condition 1.
Similarly,
\begin{equation}
\sqrt{n}(\hat{\gamma}_n-\tilde{\gamma}_0)
=
\{1+o_p(1)\}\frac{1}{\sqrt{n}}\sum_{i=1}^n B^{-1}Y_i\delta_i.
\label{eq:gamma_exp_joint}
\end{equation}

The population normal equation gives
\[
E[X\{m(X)-X\tilde{\beta}_0\}]=0,
\]
and since \(E[\eta\mid X]=0\),
\[
E[X\eta]=E[XE(\eta\mid X)]=0.
\]
Hence $E[X\epsilon]=0$. Similarly, because \(E[\xi\mid Y]=0\),
\[
E[Y\delta]=E[Y\{l(Y)-Y\tilde{\gamma}_0\}]+E[Y\xi]=0.
\]
By Condition 3, the variances
\[
A^{-1}E[X^2\epsilon^2]A^{-1}
\qquad\text{and}\qquad
B^{-1}E[Y^2\delta^2]B^{-1}
\]
are finite, so by the central limit theorem, $\sqrt{n}(\hat{\beta}_n-\tilde{\beta}_0)$ and $\sqrt{n}(\hat{\gamma}_n-\tilde{\gamma}_0)$ are asymptotically Gaussian.

We next expand \(h_{ij}=h(e_i,e_j)\) around \(h(\epsilon_i,\epsilon_j)\):
\[
h_{ij}
=
h(\epsilon_i,\epsilon_j)
+
(e_i-\epsilon_i)h_x(\lambda_{ijn},\tau_{ijn})
+
(e_j-\epsilon_j)h_y(\lambda_{ijn},\tau_{ijn}),
\]
where \((\lambda_{ijn},\tau_{ijn})\) lies on the line segment joining
\((e_i,e_j)\) and \((\epsilon_i,\epsilon_j)\). Similarly,
\[
h'_{ij}
=
h'(\delta_i,\delta_j)
+
(d_i-\delta_i)h'_x(\lambda'_{ijn},\tau'_{ijn})
+
(d_j-\delta_j)h'_y(\lambda'_{ijn},\tau'_{ijn}),
\]
where \((\lambda'_{ijn},\tau'_{ijn})\) lies on the line segment joining
\((d_i,d_j)\) and \((\delta_i,\delta_j)\).

Using
\[
e_i-\epsilon_i=-X_i(\hat{\beta}_n-\tilde{\beta}_0),
\qquad
d_i-\delta_i=-Y_i(\hat{\gamma}_n-\tilde{\gamma}_0),
\]
we obtain the decompositions
\begin{equation}
\label{eq:T_decomp_joint}
\begin{pmatrix}
T_{1b,Y,n}\\
T_{1c,X,n}
\end{pmatrix}
=
\begin{pmatrix}
T^{(0)}_{1b,Y,n}+(\hat{\beta}_n-\tilde{\beta}_0)T^{(1)}_{1b,Y,n}+R_n\\
T^{(0)}_{1c,X,n}+(\hat{\gamma}_n-\tilde{\gamma}_0)T^{(1)}_{1c,X,n}+R_n'
\end{pmatrix},
\end{equation}
where
\begin{align*}
T^{(0)}_{1b,Y,n}
&=
\frac{1}{n^2}\sum_{i,j=1}^n k(X_i,X_j)h(\epsilon_i,\epsilon_j)
+
\frac{1}{n^4}\sum_{i,j,q,r=1}^n k(X_i,X_j)h(\epsilon_q,\epsilon_r)
-
\frac{2}{n^3}\sum_{i,j,q=1}^n k(X_i,X_j)h(\epsilon_i,\epsilon_q),
\\
T^{(1)}_{1b,Y,n}
&=
\frac{1}{n^2}\sum_{i,j=1}^n k(X_i,X_j)h^{(1)}_{ij}
+
\frac{1}{n^4}\sum_{i,j,q,r=1}^n k(X_i,X_j)h^{(1)}_{qr}
-
\frac{2}{n^3}\sum_{i,j,q=1}^n k(X_i,X_j)h^{(1)}_{iq},
\\
T^{(0)}_{1c,X,n}
&=
\frac{1}{n^2}\sum_{i,j=1}^n k'(Y_i,Y_j)h'(\delta_i,\delta_j)
+
\frac{1}{n^4}\sum_{i,j,q,r=1}^n k'(Y_i,Y_j)h'(\delta_q,\delta_r) -
\frac{2}{n^3}\sum_{i,j,q=1}^n k'(Y_i,Y_j)h'(\delta_i,\delta_q),
\\
T^{(1)}_{1c,X,n}
&=
\frac{1}{n^2}\sum_{i,j=1}^n k'(Y_i,Y_j)h^{\prime(1)}_{ij}
+
\frac{1}{n^4}\sum_{i,j,q,r=1}^n k'(Y_i,Y_j)h^{\prime(1)}_{qr}
-
\frac{2}{n^3}\sum_{i,j,q=1}^n k'(Y_i,Y_j)h^{\prime(1)}_{iq},
\end{align*}
with
\[
h^{(1)}_{ij}
=
-\big\{h_x(\epsilon_i,\epsilon_j)X_i+h_y(\epsilon_i,\epsilon_j)X_j\big\} \quad \text{and} \quad h^{\prime(1)}_{ij}
=
-\big\{h'_x(\delta_i,\delta_j)Y_i+h'_y(\delta_i,\delta_j)Y_j\big\}.
\]
By Lemma~1 in the Supplement of \citet{sen_testing_2014} and Condition 4,
\[
R_n=o_p(n^{-1/2}),
\qquad
R_n'=o_p(n^{-1/2}).
\]

We now show that \(\theta(X,\epsilon)>0\) and \(\theta(Y,\delta)>0\). Since
\[
E[\epsilon\mid X]
=
m(X)-X\tilde{\beta}_0+E[\eta\mid X]
=
m(X)-X\tilde{\beta}_0,
\]
and
\[
E[\delta\mid Y]
=
l(Y)-Y\tilde{\gamma}_0+E[\xi\mid Y]
=
l(Y)-Y\tilde{\gamma}_0,
\]
Condition 5 implies that both conditional means are nonconstant with positive
probability. Hence \(X\not\perp \epsilon\) and \(Y\not\perp \delta\). Since the
kernels are characteristic, HSIC detects dependence, so $\theta(X,\epsilon)>0$ and $\theta(Y,\delta)>0.$

Next, let \(W_i=(X_i,\epsilon_i)\) and \(W_i'=(Y_i,\delta_i)\). Then
\(T^{(0)}_{1b,Y,n}\) and \(T^{(0)}_{1c,X,n}\) are V-statistics with finite second
moments and means \(\theta(X,\epsilon)\) and \(\theta(Y,\delta)\), respectively.
Therefore, by the standard asymptotic theory of nondegenerate V-statistics, there
exist mean-zero, square-integrable functions \(h^{(0)}_{Y,1}\) and
\(h^{(0)}_{X,1}\) such that
\begin{align}
\sqrt{n}\{T^{(0)}_{1b,Y,n}-\theta(X,\epsilon)\}
&=
\frac{1}{\sqrt{n}}\sum_{i=1}^n h^{(0)}_{Y,1}(W_i)+o_p(1),
\label{eq:proj_Y_joint}
\\
\sqrt{n}\{T^{(0)}_{1c,X,n}-\theta(Y,\delta)\}
&=
\frac{1}{\sqrt{n}}\sum_{i=1}^n h^{(0)}_{X,1}(W_i')+o_p(1).
\label{eq:proj_X_joint}
\end{align}

By the weak law of large numbers for V-statistics,
\begin{equation}
T^{(1)}_{1b,Y,n}\xrightarrow{p}\nu \quad \text{and} \quad
T^{(1)}_{1c,X,n}\xrightarrow{p}\phi,
\label{eq:T1X_limit_joint}
\end{equation}
for finite constants \(\nu\) and \(\phi\).

Combining \eqref{eq:T_decomp_joint}, \eqref{eq:beta_exp_joint},
\eqref{eq:gamma_exp_joint}, \eqref{eq:proj_Y_joint}, \eqref{eq:proj_X_joint}, and \eqref{eq:T1X_limit_joint}, we obtain
\begin{align*}
\sqrt{n}
\begin{pmatrix}
T_{1b,Y,n}-\theta(X,\epsilon)\\
T_{1c,X,n}-\theta(Y,\delta)
\end{pmatrix}
&=
\begin{pmatrix}
\sqrt{n}\{T^{(0)}_{1b,Y,n}-\theta(X,\epsilon)\}
+\sqrt{n}(\hat{\beta}_n-\tilde{\beta}_0)T^{(1)}_{1b,Y,n}
+\sqrt{n}R_n
\\
\sqrt{n}\{T^{(0)}_{1c,X,n}-\theta(Y,\delta)\}
+\sqrt{n}(\hat{\gamma}_n-\tilde{\gamma}_0)T^{(1)}_{1c,X,n}
+\sqrt{n}R_n'
\end{pmatrix}
\\
&=
\begin{pmatrix}
\frac{1}{\sqrt{n}}\sum_{i=1}^n\left\{
h^{(0)}_{Y,1}(W_i)+\nu A^{-1}X_i\epsilon_i
\right\}
\\
\frac{1}{\sqrt{n}}\sum_{i=1}^n\left\{
h^{(0)}_{X,1}(W_i')+\phi B^{-1}Y_i\delta_i
\right\}
\end{pmatrix}
+o_p(1).
\end{align*}

Define
\[
\psi_i=
\begin{pmatrix}
h^{(0)}_{Y,1}(W_i)+\nu A^{-1}X_i\epsilon_i\\
h^{(0)}_{X,1}(W_i')+\phi B^{-1}Y_i\delta_i
\end{pmatrix}.
\]
Then \(\psi_1,\psi_2,\ldots\) are i.i.d., \(E[\psi_i]=0\), and by Condition 3
they have finite second moments. By the multivariate central limit theorem,
\[
\frac{1}{\sqrt{n}}\sum_{i=1}^n \psi_i
\xrightarrow{d}
N_2(0,\Sigma),
\]
where $\Sigma=\Var(\psi_1)$, with entries given in \eqref{eq:var_1bY}--\eqref{eq:var_1cX}. This completes the proof.
\end{proof}

\section{Simulation Setup Details}
\label{app:sim_truth}

\paragraph{Purpose.}
To study how CDSP responds under varying degrees of misspecification of the linear model, controlled by degrees $d \in \{1,1.2,1.25,1.3,1.4,1.5,3\}$. Note that all simulation experiments were conducted using a SLURM-based cluster computing environment to efficiently handle the computational demands. 

\paragraph{Inputs.}
Fix significance level $\alpha$, number of replications $M$, dataset size $N$, and a large Monte Carlo sample size $N_{\text{MC}}$ used to approximate population quantities.

\paragraph{Global oracle objects (computed once).}
Using a very large Monte Carlo sample from the \emph{true DGP}, $Y = f(X) + \eta$, $f(x) = \beta\,\mathrm{sign}(x-a)\,|x-a|^d$, with $X\sim P_X$ (truncated exponential) and $\eta\sim P_\eta$ (mixture), compute estimates of the population HSIC quantities $\theta(Y,\delta)$ and $\theta(X,\epsilon)$ by computing the corresponding test statistics using this large Monte Carlo dataset.
These quantities are treated as oracle (population) values.

\paragraph{Per-replication experiments.}
For each misspecification setting $d$ and for $m=1,\dots,M$:
\begin{enumerate}
  \item \textbf{Simulate data from the true DGP.} Generate a dataset $\mathcal{D}^{(m)}=\{(X_i^{(m)},Y_i^{(m)})\}_{i=1}^N$ using $X\sim P_X$, $\eta\sim P_\eta$, and $Y=f(X)+\eta$.
  
  \item \textbf{LiNGAM decision.} Fit DirectLiNGAM and record the chosen direction.
  
  \item \textbf{Estimate directional detectability indices using bootstrap.} Using bootstrap resampling on $\mathcal{D}^{(m)}$, obtain bootstrap estimates of the directional detectability indices $\widehat I^{(m)}_{Y,n}$ and $\widehat I^{(m)}_{X,n}$. Use this to estimate the causal direction via effect-size asymmetry.
  
  \item \textbf{Store per-replication test statistics.} Record $T_{1b,Y}^{(m)}$ and $T_{1c,X}^{(m)}$ for use in estimating the oracle asymptotic standard deviations.
  
  
  
  
\end{enumerate}

\paragraph{Population Effect-Size Asymmetry}
For each misspecification setting, use the following oracle estimates to evaluate the population effect-size asymmetry assumption via $I_X$ and $I_Y$:
\begin{itemize}


\item \textbf{Oracle standard deviations.} Estimate
$\sigma_{1b,Y} \approx \sqrt{N}\,\mathrm{sd}_M\!\big(T_{1b,Y}^{(m)}\big), \sigma_{1c,X} \approx \sqrt{N}\,\mathrm{sd}_M\!\big(T_{1c,X}^{(m)}\big)$.

\item \textbf{Population HSIC.} Use the Monte Carlo estimates of $\theta(Y,\delta)$ and $\theta(X,\epsilon)$ computed from the global oracle objects.

\end{itemize}

\paragraph{$CDSP$ and LiNGAM Accuracy.}
Across $m=1,\dots,M$ for each misspecification setting:
\begin{itemize}
  \item Estimate $CDSP$ accuracy as the proportion of replications in which the direction favored by the bootstrap-estimated directional detectability indices $\hat{I}_{X,n}$ and $\hat{I}_{Y,n}$ matches the direction favored by the population quantities $I_X$ and $I_Y$.
  \item Estimate LiNGAM accuracy as the proportion of replications in which LiNGAM identifies the true causal direction.
\end{itemize}

\section{Severe Model Misspecification Simulation Results}
\label{app:severe_miss}
\begin{figure}[!ht]
    \centering
    \includegraphics[width=0.7\linewidth]{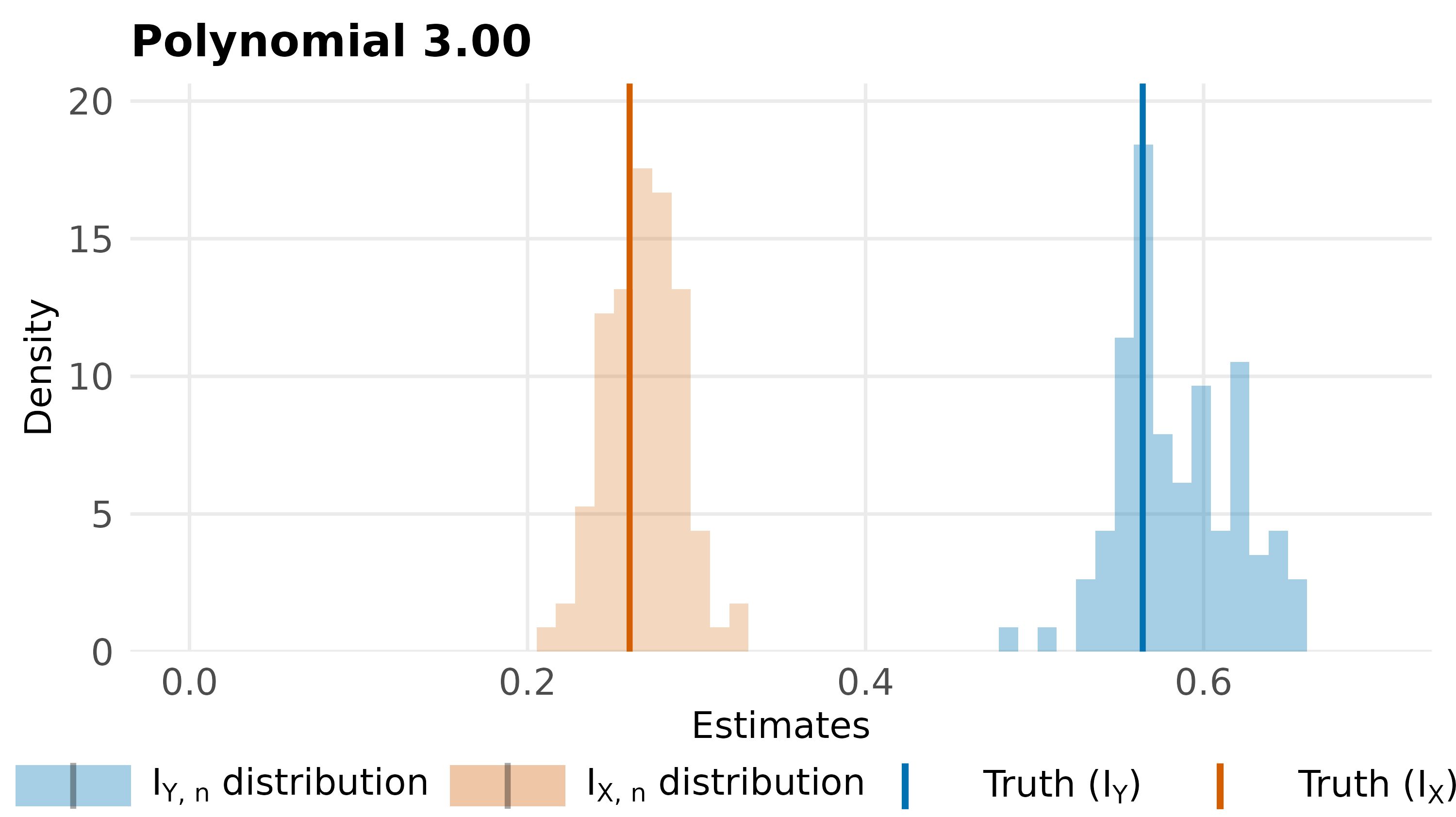}
    \caption{Distribution of the estimated directional detectability indices $I_{Y,n}$ (blue) and $I_{X,n}$ (orange) under $d = 3$. Solid vertical lines indicate the corresponding population values ($I_Y$ in blue and $I_X$ in orange).}
    \label{fig:p3_sims}
\end{figure}

We present CDSP under severe linear model misspecification with exponent $d = 3$. Under this setting, the effect-size asymmetry assumption no longer holds. The results for this setting are shown in Figure~\ref{fig:p3_sims}.

In contrast to the moderate misspecification settings, we again observe no overlap between the distributions of $I_{Y,n}$ and $I_{X,n}$, but now the ordering is reversed. Consequently, direction detection via effect-size asymmetry favors the incorrect direction $Y \to X$ across all Monte Carlo replications as the effect-size asymmetry assumption does not hold.

This severe misspecification setting highlights an important aspect of CDSP. While the method is robust to model misspecification, there exists a threshold at which the effect-size asymmetry assumption breaks down. As the degree of misspecification increases, the overlap between the $I_{Y,n}$ and $I_{X,n}$ distributions again diminishes, but in favor of the opposite direction. However, as discussed in Remark~\ref{rem:asymptotics_relevance}, this regime is not the primary setting for which CDSP is intended, since practitioners would typically avoid linear-model-based approaches under such severe misspecification.

\end{document}